\documentclass[10pt]{article}
\usepackage{amsmath,amssymb,amsthm,,amsbsy,bm,bbm,latexsym,color}
\usepackage[letterpaper,margin=1.2in]{geometry}
\usepackage[english]{babel}
\usepackage{hyperref}
\usepackage{mathrsfs}
\usepackage{todonotes}
\usepackage{mathabx}

\usepackage{changes}
\usepackage{soul}

\newtheorem{proposition}{Proposition}[section]
\newtheorem{theorem}{Theorem}[section]
\newtheorem{lemma}[theorem]{Lemma}

\newtheorem{remark}[theorem]{Remark}

\renewcommand{\Re}{\mathrm{Re}}
\renewcommand{\Im}{\mathrm{Im}}

\newcommand{\tr}{\mathrm{Tr}}

\newcommand{\id}{\mathbbm{1}}

\newcommand{\Number}{\mathcal{N}}

\newcommand{\ch}{\chi_{\le N}}
\newcommand{\chb}{\chi_{\rm{B}}}
\newcommand{\chbm}{\chi_{\rm{B},M}}

\newcommand{\abs}[1]{\left| #1 \right|}
\newcommand{\norm}[1]{\left\| #1 \right\|}
\newcommand{\scp}[2]{\left\langle #1 , #2 \right\rangle}

\newcommand{\bra}[1]{\langle #1 |}
\newcommand{\ket}[1]{| #1 \rangle}

\usepackage[normalem]{ulem}

\definechangesauthor[name={David}, color=blue]{D}
\definechangesauthor[name={Nick}, color=olive]{N}
\definechangesauthor[name={Marco}, color=orange]{M}
\definechangesauthor[name={Soeren}, color=red]{S}

\numberwithin{equation}{section}

\begin{document}

\title{Norm approximation for the Fr\"ohlich dynamics in the mean-field regime}

\author{Nikolai Leopold\footnote{University of Basel, Department of Mathematics and Computer Science, Spiegelgasse 1, 4051 Basel, Switzerland, E-mail address: {\tt nikolai.leopold@unibas.ch}} }

\maketitle

\frenchspacing

\begin{abstract}

We study the time evolution of the Fr\"ohlich Hamiltonian in a  mean-field limit in which many particles weakly couple to the quantized phonon field. Assuming that the particles are initially in a Bose-Einstein condensate and that the excitations of the phonon field are initially in a coherent state we provide an effective dynamics which approximates the time evolved many-body state in norm, provided that the number of particles is large. The approximation is given by a product state which evolves according to the Landau--Pekar equations and which is corrected by a Bogoliubov dynamics. In addition, we extend the results from \cite{LMS2021} about the approximation of the time evolved many-body state in trace-norm topology to a larger class of many-body initial states with an improved rate of convergence.

\end{abstract}

\section{Introduction and main results}

We are interested in the evolution of a Bose-Einstein condensate which weakly interacts with the excitations of a quantized phonon field. For this purpose we consider the Fr\"ohlich model in the mean-field regime. It is defined on the Hilbert space
\begin{align}
\mathcal{H}^{(N)} = \left(L^2 \left( \mathbb{R}^{3} \right) \right)^{\otimes_s N} \otimes \mathcal{F} ,
\end{align}
where $\otimes_s$ denotes the symmetric tensor product and 
$\mathcal{F}  =  \bigoplus_{n = 0}^{\infty} \left(L^2 \left( \mathbb{R}^{3} \right) \right)^{\otimes_s n}
$.
The state of system evolves according to the Schr\"odinger equation
\begin{align}
\label{eq:Schroedinger equation}
 i \partial_t \Psi_{N,t} = H^{\rm F}_{N,\alpha} \Psi_{N,t}
\end{align}
with Fr\"ohlich Hamiltonian
\begin{align}
\label{eq: Froehlich Hamiltonian quadratic form}
H^{\rm F}_{N,\alpha} &= \sum_{j=1}^N  \left[ - \Delta_j 
+ \sqrt{\frac{\alpha}{N}}
\int dk \, \abs{k}^{-1}
\left( e^{2 \pi ikx_j}  a_k + e^{- 2 \pi i k x_j } a^*_k  \right)  \right] + \mathcal{N}_a .
\end{align}
The annihilation operators $a_k$ and creation operators $a_k^*$ satisfy the canonical commutation relations
\begin{align}
\label{eq: canonical commutation relation}
[a_k, a^*_l ] &=  \delta(k-l), \quad
[a_k, a_l ] = 
[a^*_k, a^*_l ] = 0
\end{align}
and $\mathcal{N}_a$ is the number operator defined by $\mathcal{N}_a= \int d^3 k\, a_k^* a_k$. 
The coupling parameter $\sqrt{\alpha/N}$ scales the strength of the interaction. It is chosen in a way such that all terms in the Hamiltonian are of order $N$ if the number of phonons is of order $N$. Let us remark that the definition \eqref{eq: Froehlich Hamiltonian quadratic form} is rather formal since the form factor of the phonon field is not square integrable. Using the commutator method of Lieb and Yamazaki \cite{LY1958} the Hamiltonian $H_{N,\alpha}^{\rm F}$ can, however, always be defined via its associated quadratic form. More information on this account and about the domain of $H_{N,\alpha}^{\rm F}$ are given in \cite{GW2016, LS2019}. 
We are interested in the evolution of many-body initial states of the form
\begin{align} 
\label{eq:initial data many body state}
\Psi_{N} &\approx \psi^{\otimes N}\otimes W(\sqrt N \varphi) \Omega.
\end{align}
Here $\psi, \varphi \in L^2(\mathbb{R}^3)$, $\Omega$ denotes the vacuum in $\mathcal{F}$ and $W(f)$ (with $f \in L^2(\mathbb{R}^3)$) is the unitary Weyl operator
\begin{align}
\label{eq: Weyl operator}
W(f) = \exp \left(  \int dk \, \big( f(k) a^*_k - \overline{f(k)} a_k \big) \right)
\end{align} 
satisfying 
\begin{align}
\label{eq: Weyl operators product}
W^{-1}(f)= W(-f), \quad W(f) W(g) = W(g) W(f) e^{-2 i \Im \langle f , g\rangle}  = W(f+g) e^{-i \Im \langle f , g \rangle }
\end{align}
as well as the shift property
\begin{align}
\label{eq: Weyl operators shift property}
W^*(f) a_k W(f) = a_k + f(k).
\end{align}
The initial datum \eqref{eq:initial data many body state} describes a Bose-Einstein condensate of particles with condensate wave function $\psi$ and a coherent state of phonons with mean particle number $N \norm{\varphi}_{L^2(\mathbb{R}^3)}^2$. A central feature of \eqref{eq:initial data many body state} is that particles and phonons only have little correlations among each other.
During the time evolution correlations emerge because of the interaction and the state of the system will no longer be of product type. If the number of particles $N$ is large these correlations are, however, weak enough such that the solution of the many-body Schr\"odinger equation \eqref{eq:Schroedinger equation} with initial datum \eqref{eq:initial data many body state} can be approximated in trace-norm topology (see Theorem \ref{theorem:reduced density matrices} and \cite{LMS2021}) by a product state
\begin{align}
\label{eq:time evolved Pekar product state}
\psi_t^{\otimes N}\otimes W(\sqrt N \varphi_t) \Omega,
\end{align}
where $(\psi_t,\varphi_t) \in L^2(\mathbb R^3) \times L^2(\mathbb R^3)$ solves the time-dependent Landau--Pekar equations
\begin{align}
\label{eq:Landau Pekar equations}
\begin{cases}
 i \partial_t \psi_t(x) & = h(t) \psi_t(x),  \\[2mm]
i \partial_t \varphi_t(k) & = \ \   \varphi_t(k) + \sqrt{ \alpha } \abs{k}^{-1}  \int dx \,  e^{-2 \pi i k x} \abs{\psi_t(x)}^2  
\end{cases}
\end{align}
with 
\begin{align}
\label{eq:definition classical phonon field}
\Phi_{\varphi}(x) &=
 \int dk \, \abs{k}^{-1}
\left(  e^{2 \pi i k x} \varphi(k)  + e^{- 2 \pi i k x}  \overline{\varphi(k)}  \right) 
\quad \text{for} \, \varphi \in L^2(\mathbb{R}^3) , 
\\
\mu(t) &= \frac{1}{2} \sqrt{\alpha} \int d x \, \Phi_{\varphi_t}(x) \abs{\psi(t,x)}^2 
\\
h(t) &= - \Delta + \sqrt{ \alpha}  \Phi_{\varphi_t} 
 - \mu(t) 
\end{align}
and initial datum $(\psi_0,\varphi_0) = (\psi, \varphi)$. The Landau--Pekar equations model the interaction between a single quantum particle and a classical phonon field. The corresponding energy functional is given by 
\begin{align}
\label{eq:definition energy L-P}
\mathcal{E}[\psi, \varphi]
&= \scp{\psi}{\left( - \Delta + \sqrt{\alpha} \Phi_{\varphi} \right) \psi}_{L^2(\mathbb{R}^3)} + \norm{\varphi}_{L^2(\mathbb{R}^3)}^2  .
\end{align}
Approximating $\Psi_{N,t}$ by \eqref{eq:time evolved Pekar product state} enables us to compute the many-body time evolution by a set of non-linear differential equations which involve only the condensate wave function and a classical phonon field. This reduces the complexity of the description tremendously. The validity of the approximation, however, only holds for the particle and phonon reduced density matrices of $\Psi_{N,t}$.
In order to obtain an effective description which approximates $\Psi_{N,t}$ in the Hilbert space norm of $\mathcal{H}^{(N)}$ it is necessary to take the leading order of the quantum fluctuations around the Landau--Pekar equations into account. This can be established by means of a Bogoliubov dynamics which is defined on the tensor product of two Fock spaces (see Subsection \ref{subsection:Excitation Fock-space and Bogoliubov Hamiltonian}). The Bogoliubov dynamics is generated by a quadratic Hamiltonian and for this reason again much easier to analyze than the true time evolution. For a detailed discussion of this fact in the context of the Nelson model with ultraviolet cutoff we refer to \cite[Chapter 3]{FLMP2021}.

The goal of this work is twofold:
In the first part we look at the reduced density matrices of $\Psi_{N,t}$ and prove that their time evolutions are approximately given by the Landau--Pekar equations. The second part shows that the Bogoliubov dynamics from Subsection \ref{subsection:Excitation Fock-space and Bogoliubov Hamiltonian} takes the correlations between the particles and the phonons correctly into account, leading to a norm approximation of the time evolved many-body state.

The approximation for the reduced densities has previously been shown in \cite{LMS2021}. The merit of Theorem \ref{theorem:reduced density matrices} is to extend these results to a larger class of many-body initial states with an improved rate of convergence. In contrast to \cite{LMS2021} we do not require any assumptions on the variance of the many-body energy and for this reason are able to classify the many-body initial data and to prove the results without using the Gross transform. The fact that the many-body initial data is only restricted to the form domain of the non-interacting Hamiltonian is of great importance if one would like to study the mean-field limit of the renormalized Nelson model with similar techniques and a regularization of the interaction by means of the Gross transform.\footnote{With this respect it is important to note that the Gross transform classifies the domain of the Fr\"ohlich Hamiltonian while it only allows to study the form domain of the Nelson Hamiltonian \cite{GW2016, GW2018, LS2019}.}
To the best of our knowledge,
Theorem \ref{theorem:Bogoliubov approximation} provides the first rigorous derivation of the Bogoliubov dynamics for the Fr\"ohlich model in the mean-field regime.

In order to obtain our result we use the excitation map from \cite{FLMP2021} (a straightforward generalization of the excitation map originally introduced in \cite{LNSS2015}), the commutator method of Lieb and Yamazaki \cite{LY1958}, an operator bound which follows from \cite[Lemma 10]{FS2014} and energy estimates in the spirit of \cite[Theorem 8]{LNS2015} which have previously been applied in derivations of the Hartree-- and Gross--Pitaevskii equations \cite{BCS2017, BOS2015, BNNS2019, BS2019, NN2017, NN2019, NS2020}. The norm approximation of the many-body quantum state is proven by means of the approach from \cite{LNS2015}. The main difficulties arise from the singular interaction which requires to control the growth of the kinetic energy not only under the Bogoliubov dynamics but also under the fluctuations dynamics around the true many-body evolution. In addition, it is necessary to estimate the growth of higher moments of the number operator during the Bogoliubov time evolution. To this end we introduce a Bogoliubov dynamics which is truncated in the total number of particles. Such a truncation has previously been used in \cite{BNNS2019, NN2017, NN2019, NS2020, RS2009}.

\paragraph{Comparison with the literature.}
The Fr\"ohlich model with $N=1$ was originally introduced in \cite{F1937} to describe the behavior of an electron in an ionic crystal and the Landau--Pekar equations were presented in \cite{LP1948} as effective equations for this model in the strong coupling limit, i.e. $H_{1,\alpha}^{\rm F}$  with $\alpha \gg 1$. In recent years the rigorous derivation of the Landau--Pekar equations in the strong coupling regime was established in a  series of works \cite{FG2017, FS2014, G2017,LMRSS2021, LRSS2021, M2021}. For a comparison between the different results we refer to \cite[p. 658]{LMRSS2021}. In this regard let us also mention \cite{FG2020, LRSS2021, FRS2021} for results on adiabatic theorems of the Landau--Pekar equations (in one and three dimensions) and on the persistence of the spectral gap during the evolution of the Landau--Pekar equations.
In this work we are concerned with the Fr\"ohlich Hamiltonian in the mean-field regime as previously considered in \cite{LMS2021}. Here, the appearance of classical radiation rests on the fact that many weakly correlated particles in the same quantum state create radiation. This mechanism has been investigated for the Nelson model with ultraviolet cutoff \cite{AF2014, F2013, FLMP2021, LP2018}, the Pauli--Fierz Hamiltonian \cite{FL2022, LP2020}, the renormalized Nelson model \cite{AF2017} and for the Nelson model with ultraviolet cutoff in a limit of many weakly interacting fermions \cite{LP2019}. We would like to remark that the interaction of the renormalized Nelson model is more singular than the one of the Fr\"ohlich model and that an analysis  is more complicated in this case. The results from \cite{AF2017} do, however, not provide explicit error estimates and we view the present article as a starting point for a derivation of the Schr\"odinger--Klein--Gordon equations from the renormalized Nelson model with an explicit rate of convergence. Let us also mention \cite{AFH2022, CCFO2021, CF2018, CFO2019, D1979, GNV2006, T2002}  in which the classical behavior of radiation fields was shown in other scaling regimes.

The article is structured as follows:
In the rest of this section we present our main results.
Section~\ref{sec: Notation and basic estimates} specifies the notation and provides a series of operator estimates which will be useful for the proofs. The results are proven in Section~\ref{sec:proofs}.

\subsection{Approximation of the one-particle reduced density matrices}
As in \cite{LMS2021} we set the coupling constant $\alpha = 1$ and use the notation $H_{N}^{\rm F} = H_{N,1}^{\rm F}$. All results are, however, equally true for any $\alpha >0$ independent of $N$. For $m \in \mathbb{N}$, let $H^m(\mathbb{R}^3)$ denote the Sobolev space of order $m$ and $L_m^2(\mathbb{R}^3)$ be the weighted $L^2$-space with norm $\| \varphi\| _{L_m^2(\mathbb{R}^3)} = \| ( 1 + \abs{\,\cdot\,}^2 )^{m/2} \varphi \|_{L^2(\mathbb{R}^3)}$.
Within this work we rely on the following result about the Landau--Pekar equations which in a slightly different version was proven in \cite[Lemma~2.1 and Proposition~2.2]{FG2017}. In Appendix \ref{section:solution theory for LP equation} it is outlined how the proof of \cite[Proposition~2.2]{FG2017} has to be modified.
\begin{proposition}
\label{proposition:solution theory for LP equation}
For any $(\psi, \varphi) \in H^1(\mathbb{R}^3) \times L^2(\mathbb{R}^3)$ there is a unique global solution $(\psi_t, \varphi_t)$ of \eqref{eq:Landau Pekar equations}. One has the conservation laws
\begin{align}
\norm{\psi_t}_{L^2(\mathbb{R}^3)} = \norm{\psi}_{L^2(\mathbb{R}^3)}
\quad \text{and} \quad 
\mathcal{E}[\psi_t, \varphi_t] = \mathcal{E}[\psi, \varphi]
\quad \text{for all} \; t \in \mathbb{R} .
\end{align}

If $(\psi, \varphi) \in H^3(\mathbb{R}^3) \times L_{2}^2(\mathbb{R}^3)$  with $\norm{\psi}_{L^2(\mathbb{R}^3)} = 1$ then there exists a constant $C> 0$ depending only on the initial data such that 
\begin{align}
\label{eq:solution of the SKG equations bounds for the H-3 and L-2-2 norm}
\norm{\varphi_t}_{L_2^2(\mathbb{R}^3)} \leq C \left( 1 + \abs{t}^3 \right) 
\quad \text{and} \quad
\norm{\psi_t}_{H^3(\mathbb{R}^3)} \leq C \left( 1 +  \abs{t}^4  \right) 
\quad \text{for all} \; t \in \mathbb{R}.
\end{align}
\end{proposition}
Moreover, let
\begin{align}
\gamma_{\Psi_N}^{(1,0)} = \tr_{2, \ldots, N} \otimes \tr_{\mathcal{F}} \big( \ket{\Psi_N} \bra{\Psi_N} \big)
\end{align}
be the one-particle reduced density matrix on $L^2(\mathbb{R}^3)$ of $\Psi_N \in \mathcal{H}^{(N)}$ and $\tr \abs{A}$ denote the trace norm of any trace class operator $A$. Our first result is the following.

\begin{theorem}
\label{theorem:reduced density matrices}
Let $(\psi, \varphi) \in H^3(\mathbb{R}^3) \times L_2^2(\mathbb{R}^3)$ s.t. $\norm{\psi}_{L^2(\mathbb{R}^3)}= 1$, $\Psi_N \in \mathcal{D} \Big( \big( \sum_{j=1}^N - \Delta_j + \mathcal{N}_a \big)^{1/2} \Big)$ s.t. $\norm{\Psi_N}_{\mathcal{H}^{(N)}} = 1$
and define
\begin{align}
a \left[ \Psi_N, \psi \right]
&= \tr \left| \sqrt{1 - \Delta} \left( 1 - \ket{\psi} \bra{\psi} \right) \gamma_{\Psi_N}^{(1,0)}  
\left( 1 - \ket{\psi} \bra{\psi} \right) \sqrt{1 - \Delta}
\right| ,
\\
b \left[ \Psi_N, \varphi \right]
&= N^{-1} \scp{W^* (\sqrt{N} \varphi)\Psi_N}{\mathcal{N}_a W^* (\sqrt{N} \varphi)\Psi_N}_{\mathcal{H}^{(N)}} .
\end{align}
Let $(\psi_t, \varphi_t)$ be the unique solution of \eqref{eq:Landau Pekar equations} with initial datum $(\psi, \varphi)$ and $\Psi_{N,t} = e^{- i H^{\rm F}_N t} \Psi_N$.
 Then there exists a constant $C> 0$ (depending only on $\mathcal{E}[\psi, \varphi]$)  and 
 $f(t) = \int_0^t ds \, \left( \norm{\psi_s}_{H^3(\mathbb{R}^3)}^2 + \norm{\varphi_s}_{L_2^2(\mathbb{R}^3)}^2 \right)$ such that
\begin{align}
\label{eq: theorem reduced density matrices estimate 1}
N^{-1} \scp{W^* (\sqrt{N} \varphi_t)\Psi_{N,t}}{\mathcal{N}_a W^* (\sqrt{N} \varphi_t) \Psi_{N,t}}_{\mathcal{H}^{(N)}}
&\leq \left( a \left[ \Psi_N, \psi \right]
+  b \left[ \Psi_N, \varphi \right] + N^{-1} \right)
C  e^{C f(t)} ,
\\
\label{eq: theorem reduced density matrices estimate 2}
\tr \left| \sqrt{1 - \Delta} \big( \gamma_{\Psi_{N,t}}^{(1,0)}  - \ket{\psi_t} \bra{\psi_t} \big)  \sqrt{1 - \Delta} \right|
&\leq \sup_{j=1,2} \left( a \left[ \Psi_N, \psi \right]
+  b \left[ \Psi_N, \varphi \right] + N^{-1} \right)^{\frac{j}{2}} C  e^{C f(t)} .
\end{align}
\end{theorem}

\begin{remark}
Note that $\lim_{N \rightarrow \infty} \left(  a \left[ \Psi_N, \psi \right] + b \left[ \Psi_N, \varphi \right] \right) = 0$ if initially the particles exhibit Bose-Einstein condensation, the kinetic energy of the particles outside the condensate is small compared to one and the phonons are in a coherent state. This implies that the left hand sides of \eqref{eq: theorem reduced density matrices estimate 1} and \eqref{eq: theorem reduced density matrices estimate 2} converge to zero as the number of particles is getting large; showing the stability of the condensate and the coherent structure during the time evolution. For Pekar product states of the form $\Psi_N = \psi^{\otimes N} \otimes W (\sqrt{N} \varphi) \Omega$ we get
\begin{align}
N^{-1} \scp{W^* (\sqrt{N} \varphi_t)\Psi_{N,t}}{\mathcal{N}_a W^* (\sqrt{N} \varphi_t) \Psi_{N,t}}_{\mathcal{H}^{(N)}} &\leq C N^{-1}  e^{C f(t)} ,
\nonumber \\
\tr \left| \sqrt{1 - \Delta} \big( \gamma_{\Psi_{N,t}}^{(1,0)}  - \ket{\psi_t} \bra{\psi_t} \big)  \sqrt{1 - \Delta} \right| &\leq C N^{-1/2}  e^{C f(t)} .
\end{align}
\end{remark}

\begin{remark}
A similar result has previously been proven in \cite{LMS2021}. In contrast to \cite{LMS2021} we do not have to ensure that the variance of $N^{-1} H_N^{\rm F}$ w.r.t. $\Psi_N$ (instead only that the kinetic energy of the particles outside the condensate) is small compared to one. We can consequently consider many-body initial states in the form domain of $H_N^{\rm F}$ and for this reason do not have to distinguish between initial states in the operator domain of $H_N^{\rm{F}}$ and those in the domain of the free Hamiltonian without interaction (as it was done in \cite[Theorem 2.1 and Theorem 2.2]{LMS2021}). In comparison to \cite{LMS2021} we obtain a better rate of convergence but a more rapid growth of the error in time. The different behavior in $t$ stems from the fact that we control the kinetic energy of the particles and the interaction by an almost conserved quantity instead of the variance of the energy per particle; which is a constant of motion due to energy conservation. Finally, let us remark that our result is slightly stronger than \cite[Theorem 2.1 and Theorem 2.2]{LMS2021} in the sense that the convergence of the one-particle reduced density matrix to the projector onto the condensate wave function holds in Sobolev and not only in trace norm. The control of the kinetic energy of the particles outside the condensate is an important ingredient in the proof of Theorem \ref{theorem:Bogoliubov approximation}.
\end{remark}

\subsection{Excitation Fock-space and Bogoliubov Hamiltonian}

\label{subsection:Excitation Fock-space and Bogoliubov Hamiltonian}

Next, we provide an effective description which approximates solutions of the many-body Schr\"odinger equation \eqref{eq:Schroedinger equation} in the Hilbert space norm. This is obtained by modifying the Pekar product state \eqref{eq:time evolved Pekar product state} with the help of a Bogoliubov dynamics. The Bogoliubov dynamics describes correlations among the particles and the phonons but also between the particles and the phonons. For this reason it is convenient to describe the state of the particles on a Fock space by itself and to define the Bogoliubov dynamics on the tensor product of this space and the Fock space of the phonon field. The appearance of the Bogoliubov dynamics is motivated by means of the strategy from \cite{LNS2015, LNSS2015}. We particularly follow the route from \cite{FLMP2021} which considers the Bogoliubov dynamics and higher order corrections of the Nelson model with ultraviolet cutoff. This means that we will factor out the condensate and the coherent state from the many-body state and study the quantum fluctuations around the mean-field dynamics.
To this end we define the excitation Fock space of the particles
\begin{align}\label{eq: definition excitation Fock space particles}
\mathcal{F}_{b,\psi_t}  = \bigoplus_{k=0}^\infty \mathcal F_{b,\psi_t}^{(k)},  \quad \text{where}\quad  \mathcal{F}_{b}^{(k)} =  \left(  L^2_{\perp \psi_t}(\mathbb{R}^3) \right)^{\otimes_s k}
\end{align}
--with $ L^2_{\perp \psi_t}(\mathbb{R}^3)$ being the orthogonal complement of the one-dimensional space spanned by $\psi_t$ in $L^2(\mathbb{R}^3)$--
as well as the excitation Fock space of the phonons\footnote{Note that $\mathcal{F}_b = \mathcal{F}_a = \mathcal{F}$ because the Weyl operator maps $\mathcal{F}$ into itself. The notations $\mathcal{F}_b$ and $\mathcal{F}_a$ are introduced to distinguish between the Fock spaces of the particles and phonon field as well as to stress that the new vacuum $\Omega_a$ of $\mathcal{F}_a$ is given by $W \left( \sqrt{N} \varphi_t \right) \Omega$ where $\Omega$ is the vacuum of $\mathcal{F}$.}
\begin{equation}
\label{eq: definition excitation Fock space field bosons} 
   \mathcal{F}_{a}  =  W^* \big(\sqrt{N} \varphi_t \big)  \mathcal{F} .
\end{equation}
The truncated excitation Fock space $\mathcal G^{\leq N}_{\perp \psi_t}$ and the excitation Fock space $\mathcal G_{\perp \psi_t}$ are then given by
\begin{align}
\label{eq:double Fock space}
\mathcal G^{\leq N}_{\perp \psi_t} =  \Big( \bigoplus_{k=0}^N \mathcal F_{b,\psi_t}^{(k)}  \Big) \otimes \mathcal F_a 
=  \bigoplus_{k=0}^N \Big( \mathcal F_{b,\psi_t}^{(k)}  \otimes \mathcal F_a \Big) 
\; \subset  \;
\mathcal G_{\perp \psi_t} = \mathcal{F}_{b,\psi_t}  \otimes   \mathcal{F}_{a}
\; \subset  \;
\mathcal{G} = \mathcal{F}_b \otimes \mathcal{F}_a 
\end{align}
with $\mathcal{F}_b = \bigoplus_{n=0}^{\infty} \left(L^2 \left( \mathbb{R}^{3} \right) \right)^{\otimes_s n}$.
For $\mathcal G^{\leq N}_{\perp \psi_t}$, $\mathcal G_{\perp \psi_t}$ and $\mathcal{G}$ we equally use $b_x$, $b_x^*$, $\mathcal{N}_b$ and $a_k$, $a_k^*$, $\mathcal{N}_a$ to denote the annihilation operator, the creation operator and the number operator of the particles and the phonons respectively.
For given $\psi_t, \varphi_t \in L^2(\mathbb{R}^3)$ such that $\norm{\psi_t}_{L^2(\mathbb{R}^3)} = 1$ the unitary mapping $U_N(t) :\mathcal{H}^{(N)} \rightarrow \mathcal G^{\leq N}_{\perp \psi_t}$, $\Psi_{N} \mapsto \left( \ch^{(k)}(t) \right)_{k=0}^N$ 
with 
\begin{align}
\label{eq:action of the unitary}
\ch^{(k)}(t) = \binom{N}{k}^{1/2} \prod_{i=1}^k q_i(t)
\scp{\psi_t^{\otimes (N-k)}}{W^* \left( \sqrt{N} \varphi_t \right) \Psi_N}_{L^2(\mathbb{R}^{3(N-k)})} 
\; \in \mathcal{F}_{b, \psi_t}^{(k)} \otimes \mathcal{F}_a
\end{align}
factors out the condensate wave function $\psi_t$ and the coherent state with mean particle number $N \norm{\varphi_t}^2_{L^2(\mathbb{R}^2)}$. Here, $p_i(t) =  \ket{\psi_t} \bra{\psi_t}_{i}$ projects the coordinate of the i-th particle onto $\psi_t$, $q_i(t) = 1 - p_i(t)$ and the partial inner product is taken w.r.t. the coordinates $x_{k+1}, \ldots, x_N$ of the particles. The inverse of $U_N(t)$ is given by
\begin{align} 
\label{eq:decomposition many-body state}
\left( \ch^{(k)}(t) \right)_{k=0}^N \mapsto
\Psi_{N} = W\big(\sqrt{N} \varphi_t \big) \sum_{k=0}^N \psi_t^{\otimes (N-k)} \otimes_s \chi_{\le N}^{(k)}(t)  .
\end{align}
Note that the definition of $U_N(t)$ is a straightforward generalization of the unitary map originally introduced in \cite{LNSS2015}. A more detailed introduction of $U_N(t)$ and its properties is given in \cite[Appendix A]{FLMP2021} . If $\Psi_{N,t}$ satisfies  \eqref{eq:Schroedinger equation} and $(\psi_t, \varphi_t)$ is a solution of the Landau--Pekar equations \eqref{eq:Landau Pekar equations} the time evolution of $\ch(t)$ is given by the Schr\"odinger equation
\begin{equation}
\label{eq:Schroedinger equations in double excitation space}
i\partial_t \chi_{\le N}(t) = H^{\leq N}(t) \chi_{\le N }(t) 
\end{equation}
with Hamiltonian $H^{\leq N}(t) = U_N(t) H_{N,\alpha}^{\rm F} U_N(t)^* + i \dot{U}_{N}(t) U_N(t)^*$. For later purposes it is convenient to define $H^{\leq N}(t) = H(t) |_{\mathcal G^{\leq N}_{\perp \psi_t}}$ as a Hamiltonian $H(t)$ on $\mathcal{G}$ whose action is restricted to elements of $\mathcal G^{\leq N}_{\perp \psi_t}$. Note that the definition of $H(t)$ is not unique since the orthogonal projection from $\mathcal{G}$ to $\mathcal G^{\leq N}_{\perp \psi_t}$ has a nontrivial kernel. We set $\alpha = 1$ and use the definition (see \cite[Appendix A]{FLMP2021} for a detailed derivation in case of the Nelson model with ultraviolet cutoff)
\begin{align}
\label{eq:H on the double Fock space}
\begin{split}
H(t) &= \int dx\, b^{*}_x h(t) b_x + \mathcal{N}_a \\
&\quad + \int dx \int dk\,  K(t,k,x) \big( a^{*}_k + a_{-k}\big) b^*_x \big[ 1- N^{-1}\Number_b \big]_+^{1/2} + \text{h.c.}  \\
&\quad + N^{-1/2} \int dx\, b^{*}_x \Big( q(t) \widehat{\Phi} q(t) - \scp{\psi_t}{\widehat{\Phi} \psi_t}_{L^2(\mathbb{R}^3)} \Big) b_x,
\end{split}
\end{align}
where
\begin{align}
\widehat{\Phi}(x) &=  \int dk \, \abs{k}^{-1}
\left( e^{2 \pi ikx}  a_k + e^{- 2 \pi i k x } a^*_k  \right)
\end{align}
is interaction of the Fr\"ohlich Hamiltonian, $q(t) = 1 - \ket{\psi_t} \bra{\psi_t}$ is a projection on $L^2(\mathbb{R}^3)$ with integral kernel
\begin{align}\label{Kernel_q}
q(t,x,y) = \delta(x-y) - \psi_t(x) \overline{\psi_t(y)} 
\end{align}
and 
\begin{align}
K(t,k,x) &= \int dy \, q(t,x, y) \abs{k}^{-1} e^{- 2 \pi i k y} \psi_t(y) .
\end{align}
We use $[x]_+$ to denote the positive part of $x$ and $\text{h.c.}$ to indicate the Hermitian conjugate of the preceding term. Moreover, note that $\int dx \, b_x^* A b_x = \int dx \, \int dy \, b_x^* A(x;y) b_y$ is the usual shorthand notation for operators on $L^2(\mathbb{R}^3)$ with integral kernel $A(x;y)$.
Disregarding all terms of $H(t)$ with more than two annihilation or creation operators leads to the Bogoliubov Hamiltonian
\begin{align}
\label{eq:Bogoliubov Hamiltonian definition}
\begin{split}
H^{\rm{B}} (t) &= \int dx\, b^{*}_x h(t)  b_x + \mathcal{N}_a 
+ \left( \int dx \int dk\,  K(t,k,x) \big( a_k^{*} + a_{-k} \big) b^*_x + \text{h.c.}  \right)  .
\end{split}
\end{align}
The Bogoliubov equation is given by
\begin{align}
\label{eq:Bogoliubov dynamics}
i\partial_t \chb(t) = H^{\rm{B}} (t) \chb(t) 
\quad \text{with} \quad  \chb(0) \in \mathcal{G} .
\end{align}
Note that the Bogoliubov Hamiltonian does, in contrast to $H(t)$, not map $\left\{ \chi \in \mathcal{G} :  \id_{\mathcal{N}_b > N} \chi = 0 \right\}$ into itself. 
 For this reason it is impossible to define the Bogoliubov equation on the truncated Fock space.
The well-posedness of \eqref{eq:Bogoliubov dynamics} and the fact that $\chi \in \mathcal G_{\perp \psi_t}$ implies $\chb(t) \in \mathcal G_{\perp \psi_t}$ for all $t \in \mathbb{R}$ are addressed in Lemma \ref{lemma:well-posedness Bogoliubov dynamics}.

In the following, we will indicate elements of $\mathcal{F}_{b,\psi_t}^{(k)} \otimes \mathcal{F}_a$ and $\mathcal{F}_{b}^{(k)} \otimes \mathcal{F}_a$ by the superscript ${(k)}$.
For example, we will use $\chi^{(k)}$ to denote the sector with exactly $k$ particle excitations of a state $\chi$ in $\mathcal G^{\leq N}_{\perp \psi_t}$, $\mathcal G_{\perp \psi_t}$ and $\mathcal{G}$.

\subsection{Norm approximation of the many-body state}

\begin{theorem}
\label{theorem:Bogoliubov approximation}
Let $(\psi, \varphi) \in H^3(\mathbb{R}^3) \times L_2^2(\mathbb{R}^3)$ such that $\norm{\psi}_{L^2(\mathbb{R}^3)}= 1$ and let $(\psi_t, \varphi_t)$ be the solution of \eqref{eq:Landau Pekar equations} with initial datum $(\psi, \varphi)$. Let $T_b = - \int dx \, \int dy \, b_x^* \Delta(x;y) b_y$, $\widetilde{C} >0$ and $\chi \in \mathcal G_{\perp \psi_0}$ satisfying $\norm{\chi}_{\mathcal{G}} = 1$ as well as
\begin{align}
\label{eq:initial condition moment assumption}
\big\| ( \mathcal{N}_a^{3} + \mathcal{N}_b^{3} + T_b  )^{1/2} \chi \big\|_{\mathcal{G}} \leq \widetilde{C}.
\end{align}
Let $\Psi_{N,t}$ be the solution of the Schr\"odinger equation \eqref{eq:Schroedinger equation} with initial datum 
\begin{align}
\label{eq: norm approximation initial condition Schroedinger equation}
\Psi_N = W\big(\sqrt{N} \varphi \big) \sum_{k=0}^N \psi^{\otimes (N-k)} \otimes_s \chi^{(k)}
\in \mathcal{H}^{(N)} .
\end{align}
Then, there exists a constant $C> 0$ (depending only on $\widetilde{C}$ and $\mathcal{E}[\psi, \varphi]$) such that 
\begin{align}
\label{eq:norm approximation main result}
\norm{\Psi_{N,t} - \Psi_{N,t}^{\rm{B}}}_{\mathcal{H}^{(N)}} \leq C e^{C f(t)} N^{-1/8}  .
\end{align}
Here, $f(t) = \int_0^t ds \, \left( \norm{\psi_s}_{H^3(\mathbb{R}^3)}^2 + \norm{\varphi_s}_{L_2^2(\mathbb{R}^3)}^2 \right)$ and
\begin{align}
\label{eq:norm approximation main result approximating many-body state}
\Psi_{N,t}^{\rm{B}} = W\big(\sqrt{N} \varphi_t \big) \sum_{k=0}^N \psi_t^{\otimes (N-k)} \otimes_s \chb^{(k)}(t)  ,
\end{align}
where $\chb(t)$ is the solution of \eqref{eq:Bogoliubov dynamics} with initial datum $\chi$.
\end{theorem}

\begin{remark}
If we choose $\chi$ to be the vacuum of $\mathcal{G}$ we obtain a many-body initial state $\Psi_N = \psi^{\otimes N} \otimes W (\sqrt{N} \varphi) \Omega$ of Pekar product type.
\end{remark}

\begin{remark} Since $\chi = (\chi^{(k)})_{k\ge 0}$ is normalized to one, $\Psi_N$ is not necessarily normalized.  Assumption \eqref{eq:initial condition moment assumption} and $\norm{\psi}_{L^2(\mathbb{R}^3)} = 1$, however, imply that
\begin{align}
\norm{\Psi_{N}}_{\mathcal{H}^{(N)}}^2 = \norm{(\chi^{(k)})_{k=0}^N}_{\mathcal{G}_{\leq N}}^2
= \norm{\chi}_{\mathcal{G}}^2
- \norm{\id_{\mathcal{N}_b > N} \chi}^2_{\mathcal{G}} \rightarrow 1
\quad \text{as} \; N \rightarrow \infty
\end{align}
because $\norm{\id_{\mathcal{N}_b > N} \chi}_{\mathcal{G}} \leq N^{-3/2} \norm{\mathcal{N}_b^{3/2} \chi}_{\mathcal{G}}$.
\end{remark}

\section{Preliminaries}
\label{sec: Notation and basic estimates}

\subsection{Notation}
We introduce the usual bosonic creation and annihilation operators ($f \in L^2(\mathbb{R}^3)$)
\begin{align}
\begin{split}
a(f) &= \int d^3k  \, \overline{f(k)} a_k, \quad
a^*(f) = \int d^3k \, f(k) a^*_k , 
\\
b(f) &= \int d^3x  \, \overline{f(x)} b_x, \quad
b^*(f) = \int d^3x \, f(x) b^*_x .
\end{split}
\end{align}
They satisfy the well-known inequalities
\begin{align}
\begin{split}
\label{eq: bound for annihialtion and creation operators}
\norm{a(f) \chi}_{\mathcal{G}} & \leq \norm{f}_{L^2(\mathbb R^3)} \norm{\mathcal{N}_a^{1/2} \chi}_{\mathcal{G}} , \quad 
\norm{a^*(f) \chi}_{\mathcal{G}} \leq \norm{f}_{L^2(\mathbb R^3)} \norm{\left(\mathcal{N}_a + 1 \right)^{1/2} \chi}_{\mathcal{G}} ,
\\
\norm{b(f) \chi}_{\mathcal{G}} & \leq \norm{f}_{L^2(\mathbb R^3)} \norm{\mathcal{N}_b^{1/2} \chi}_{\mathcal{G}} , \quad 
\norm{b^*(f) \chi}_{\mathcal{G}} \leq \norm{f}_{L^2(\mathbb R^3)} \norm{\left(\mathcal{N}_b + 1 \right)^{1/2} \chi}_{\mathcal{G}} ,
\end{split}
\end{align}
for any $\chi \in \mathcal{G}$. 
We, moreover, define the total number of excitations operator 
\begin{align}
\Number = \Number_a  + \Number_b 
\end{align}
and recall the second quantization of the particle's kinetic energy
\begin{align}
T_b &= - \int dx \, b_x^* \Delta b_x .
\end{align}
In addition, it is convenient to introduce
\begin{align}
\label{eq:definition G-x}
G_{x}(k) &= e^{- 2 i \pi k x} \abs{k}^{-1} ,
\end{align}
allowing to write the interaction of the Fr\"ohlich Hamiltonian as
$\widehat{\Phi}(x) = a \left( G_x \right) + a^* \left( G_x \right)$. It holds that
\begin{align}
\label{eq:bound for G}
\sup_{x \in \mathbb{R}^3} \norm{\id_{\abs{\cdot} \leq \Lambda} G_x}_{L^2(\mathbb{R}^3)}
&= \norm{\id_{\abs{\cdot} \leq \Lambda} \abs{\cdot}^{-1}}_{L^2(\mathbb{R}^3)}
=\sqrt{4 \pi \Lambda} 
\end{align}
and 
\begin{align}
\label{eq:estimate expectation value of G}
\abs{\scp{\psi}{G_{\cdot}(k) \tilde{\psi}}_{L^2(\mathbb{R}^3)}}
&\leq   \frac{1 + \abs{k}}{\abs{k} (1 + k^2)}
\left( \norm{\psi}_{H^1(\mathbb{R}^3)}  \| \tilde{\psi} \|_{L^2(\mathbb{R}^3)}
+ \norm{\psi}_{L^2(\mathbb{R}^3)} \| \tilde{\psi} \|_{H^1(\mathbb{R}^3)} \right) ,
\end{align}
where the latter is obtained by means of \eqref{eq:commutator with G-x} and integration by parts.

The norm and scalar product of $\mathcal{G}$ will be denoted by $\norm{\cdot}$ and $\scp{\cdot}{\cdot}$. The symbol $\tr$ is used to denote the trace over $L^2(\mathbb{R}^3)$ and the trace norm of a trace class operator $\omega: L^2(\mathbb{R}^3) \rightarrow L^2(\mathbb{R}^3)$ is written as $\tr \abs{\omega}$. We use the notation $\dot{\psi}_t$ to denote the derivative of a function $\psi_t$ with respect to time. Moreover, recall that $H^m(\mathbb{R}^3)$ with $m \in \mathbb{N}$ denotes the Sobolev space of order $m$ and $L_m^2(\mathbb{R}^3)$ is a weighted $L^2$-space with norm $\| \varphi\| _{L_m^2(\mathbb{R}^3)} = \| ( 1 + \abs{\,\cdot\,}^2 )^{m/2} \varphi \|_{L^2(\mathbb{R}^3)}$.

\subsection{Interaction and Hamiltonian estimates}

In this section we provide preliminary estimates which are needed to prove the main results. We start with the interaction terms of the Hamiltonian $H(t)$. The part with two annihilation and creation operators can be controlled by  the following bounds.
\begin{lemma}
\label{lemma:preliminary estimates}
Let $\Lambda \geq 1$, $\Psi, \chi \in \mathcal{G}$, $\psi_t \in H^1(\mathbb{R}^3)$ such that $\norm{\psi_t}_{L^2(\mathbb{R}^3)} = 1$. Then, 
\begin{align}
\label{eq:preliminary estimates quadratic term 1}
&\abs{\scp{\chi}{\int dx \, \int_{\abs{k} \geq \Lambda} dk \, \psi_t(x) G_x(k) \left( a_k^* + a_{-k} \right) b_x^* \Psi}}
\leq \frac{C}{\sqrt{\Lambda}} \norm{\psi_t}_{H^1(\mathbb{R}^3)} 
\norm{\left( \mathcal{N}_b + T_b \right)^{1/2} \chi} \norm{\left( \mathcal{N}_a + 1 \right)^{1/2} \Psi} ,
\\
\label{eq:preliminary estimates quadratic term 2}
&\abs{\scp{\chi}{\int dx \, \int_{\abs{k} \leq \Lambda} dk \, \psi_t(x) G_x(k) \left( a_k^* + a_{-k} \right) b_x^* \Psi} }
\leq C \norm{\mathcal{N}_b^{1/2} \chi} 
\Big[
\sqrt{\Lambda} \norm{\Psi} + \norm{\psi_t}_{H^1} \norm{\mathcal{N}_a^{1/2} \Psi}
\Big] ,
\\
\label{eq:preliminary estimates quadratic term 3}
&\abs{\scp{\chi}{ \int dx \, \int dk \, K(t,k,x) \left(a^*_k + a_{-k} \right) b^*_x  \Psi}}
\leq C \norm{\psi_t}_{H^1(\mathbb{R}^3)}
\norm{\left( \mathcal{N}_b + T_b \right)^{1/2} \chi}
\norm{\left( \mathcal{N}_a + 1 \right)^{1/2} \Psi} .
\end{align}
\end{lemma}
The cubic interaction term of $H(t)$ is treated by means of
\begin{lemma}
\label{lemma: cubic interaction}
Let $\Lambda \geq 1$, $\Psi, \chi \in \mathcal{G}$, $\psi_t \in H^1(\mathbb{R}^3)$ such that $\norm{\psi_t}_{L^2(\mathbb{R}^3)} = 1$. Then, 
\begin{align}
\label{eq:preliminary estimate cubic term 1}
&\abs{\scp{\chi}{\int dx\, b^{*}_x \Big( q(t) \widehat{\Phi} q(t) - \scp{\psi_t}{\widehat{\Phi} \, \psi_t}_{L^2(\mathbb{R}^3)} \Big) b_x \Psi}}
\nonumber \\
&\quad \leq 
C \norm{\psi_t}_{H^1(\mathbb{R}^3)}
\left( 
\Lambda^{\frac{1}{2}} \norm{\left( \mathcal{N}_b + T_b \right)^{\frac{1}{2}} \chi}
\norm{\left( \mathcal{N}_a + 1 \right)^{\frac{1}{2}} \mathcal{N}_b^{\frac{1}{2}} \Psi}
+ \Lambda^{- \frac{1}{2}} \norm{\left( \mathcal{N}_a + 1 \right)^{\frac{1}{2}} \mathcal{N}_b^{\frac{1}{2}} \chi} \norm{\left( \mathcal{N}_b + T_b \right)^{\frac{1}{2}} \Psi}
\right) .
\end{align}
For $\varepsilon >0$ this implies
\begin{align}
\label{eq:preliminary estimate cubic term 2}
\pm  \int dx\, b^{*}_x \Big( q(t) \widehat{\Phi} q(t) - \scp{\psi_t}{\widehat{\Phi} \psi_t}_{L^2(\mathbb{R}^3)} \Big) b_x
&\leq \varepsilon \left( \mathcal{N}_b + T_b \right) 
+  C \varepsilon^{-1} \norm{\psi_t}_{H^1(\mathbb{R}^3)}^2 \left( \mathcal{N}_a + 1 \right) \mathcal{N}_b .
\end{align}
\end{lemma}
The proofs of Lemma \ref{lemma:preliminary estimates} and Lemma \ref{lemma: cubic interaction} are given in Appendix \ref{appendix:estimates for the interaction}. Inequalities \eqref{eq:preliminary estimates quadratic term 1} and \eqref{eq:preliminary estimate cubic term 1} are obtained by means of the commutator method of Lieb and Yamazaki \cite{LY1958}. Estimate \eqref{eq:preliminary estimates quadratic term 2} is derived with the help of \cite[Lemma 10]{{FS2014}}. The advantage in comparison to estimates by means of the commutator method of Lieb and Yamazaki is that the kinetic energy of the particles, $T_b$, does not appear in \eqref{eq:preliminary estimates quadratic term 2}. The two remaining inequalities follow almost immediately from the previous estimates.

By means of Lemma \ref{lemma:preliminary estimates} we obtain the following estimates for $H^{\rm B}(t)$ which will later be used to prove the well-posedness  of the Bogoliubov equation.

\begin{lemma}
\label{lemma:estimates H-0(t)}
Let $\varepsilon >0$ be arbitrary. Then, there exists a constant $C_\varepsilon > 0$, depending on $\varepsilon$, such that  
\begin{align}
\label{eq:estimates H-0(t) 1}
\pm i \left[ \mathcal{N} , H^{\rm{B}}(t) \right] 
&\leq \varepsilon \,  T_b + C_{\varepsilon} \norm{\psi_t}_{H^1(\mathbb{R}^3)}^2 \left( \mathcal{N} + 1  \right) ,
\\
\label{eq:estimates H-0(t) 2}
\pm \left( H^{\rm{B}}(t) - T_b \right)
&\leq \varepsilon T_b + C_{\varepsilon} 
\left( \norm{\psi_t}_{H^1(\mathbb{R}^3)}^2 + \norm{\varphi_t}_{L^2(\mathbb{R}^3)}^2 \right)
\left( \mathcal{N} + 1  \right) ,
\\
\label{eq:estimates H-0(t) 3}
\pm \frac{d}{dt} H^{\rm{B}}(t)
&\leq \varepsilon T_b + C_{\varepsilon} \left( \norm{\psi_t}_{H^3(\mathbb{R}^3)}^2 + 
\norm{\psi_t}_{H^1(\mathbb{R}^3)}^2 \norm{\varphi_t}_{L_{2}^2(\mathbb{R}^3)}^2 \right) 
\left( \mathcal{N} + 1  \right) .
\end{align}
\end{lemma}

\begin{proof}
By the shifting property of the annihilation and creation operators we have 
\begin{align}
 i \left[ \mathcal{N} , H^{\rm{B}}(t) \right]
=  2 i \int dx \int dk \, K(t,k,x) a_k^* b_x^* + \text{h.c.} .
\end{align}
By the same estimates as in the proof of Lemma \ref{lemma:preliminary estimates} we get \eqref{eq:estimates H-0(t) 1}. Note that 
\begin{align}
\abs{\mu(t)} = \frac{1}{2} \abs{\scp{\psi_t}{\Phi_{\varphi_t} \psi_t}_{L^2(\mathbb{R}^3)}} \leq C \left( \norm{\psi_t}_{H^1(\mathbb{R}^3)}^2 + \norm{\varphi_t}_{L^2(\mathbb{R}^3)}^2 \right) 
\end{align}
because of \eqref{eq:estimate expectation value of G}.
Using \eqref{eq:commutator with G-x}  we get
\begin{align}
\int dx \, b_x^* \Phi_{\varphi_t}(x) b_x
&= 2 \int dx \, b_x^* \Re \left\{ \int dk \, \frac{1}{1 + k^2} G_x(k) \overline{\varphi_t(k)}  \right\} b_x
\nonumber \\
&\quad -
\int dx \, b_x^* \Re \left\{ \int dk \, \frac{i k}{\pi (1 + k^2)} G_x(k) \overline{\varphi_t(k)}  \right\} i \nabla_x b_x 
+ \text{h.c.} \, .
\end{align}
By the Cauchy--Schwarz inequality and Young's inequality for products we obtain
\begin{align}
\pm \left( \int dx \, b_x^* h(t) b_x - T_b \right) \leq  \varepsilon T_b + C_{\varepsilon} \left( \norm{\psi_t}_{H^1(\mathbb{R}^3)}^2 + \norm{\varphi_t}_{L^2(\mathbb{R}^3)}^2 \right) \mathcal{N}_b .
\end{align}
Inequality \eqref{eq:estimates H-0(t) 2} then follows from Lemma \ref{lemma:preliminary estimates}.
Note that $\norm{\frac{d}{dt} \Phi(\cdot,t)}_{L^{\infty}(\mathbb{R}^3)} \leq C \norm{\varphi_t}_{L_1^2(\mathbb{R}^3)}$ because $\frac{d}{dt} \Phi_{\varphi_t}(x) = 2 \Im \scp{G_x}{\varphi_t}_{L^2(\mathbb{R}^3)}$. Together with $\norm{h(t) \psi_t}_{L^2(\mathbb{R}^3)} \leq \norm{\psi_t}_{H^2(\mathbb{R}^3)} + C \norm{\varphi_t}_{L_1^2(\mathbb{R}^3)}$ we get
$\norm{\frac{d}{dt} h(t)}_{L^\infty(\mathbb{R}^3)} \leq C \left( \norm{\psi_t}_{H^2(\mathbb{R}^3)}^2 + \norm{\varphi_t}_{L_1^2(\mathbb{R}^3)}^2 \right) $ and
\begin{align}
\pm \frac{d}{dt} \int dx \, b_x^* h(t) b_x \leq C \left(  \norm{\psi_t}_{H^2(\mathbb{R}^3)}^2 + \norm{\varphi_t}_{L_1^2(\mathbb{R}^3)}^2 \right) \mathcal{N}_b .
\end{align}
Using
\begin{align}
\frac{d}{dt} K(t,k,x)
&= \dot{\psi_t}(x) \left( G_x(k) - \scp{\psi_t}{G_x(k) \psi_t}_{L^2(\mathbb{R}^3)} \right) 
\nonumber \\
&\quad - \psi_t \left( \scp{\psi_t}{G_{\cdot}(k) \dot{\psi_t}}_{L^2(\mathbb{R}^3)}
+ \scp{\dot{\psi_t}}{G_{\cdot}(k) 
\psi_t}_{L^2(\mathbb{R}^3)} \right)
\end{align}
and similar estimates as in the proof of Lemma \ref{lemma:preliminary estimates} we obtain 
\begin{align}
&\pm \frac{d}{dt} \left( \int dx \int dk \, K(t,k,x) \left( a_k^* + a_{-k} \right) b_x^* + \text{h.c.} \right)
\nonumber \\
&\quad \leq \varepsilon T_b + C_{\varepsilon} 
\left( \| \dot{\psi_t} \|_{H^1(\mathbb{R}^3)}^2 + \| \dot{\psi_t} \|_{L^2(\mathbb{R}^3)} \norm{\psi_t}_{H^1(\mathbb{R}^3)} \right) \left( \mathcal{N} + 1 \right) .
\end{align}
Since
\begin{align}
\label{eq:norm bound for the derivative of psi-t}
\| \dot{\psi_t} \|_{L^2(\mathbb{R}^3)}
&\leq  \norm{\psi_t}_{H^2(\mathbb{R}^3)} + C \norm{\varphi_t}_{L_1^2(\mathbb{R}^3)} 
\quad \text{and} \quad
\| \dot{\psi_t} \|_{H^1(\mathbb{R}^3)}
\leq \norm{\psi_t}_{H^3(\mathbb{R}^3)} + C \norm{\psi_t}_{H^1(\mathbb{R}^3)} \norm{\varphi_t}_{L_2^2(\mathbb{R}^3)}
\end{align}
this proves \eqref{eq:estimates H-0(t) 3}.
\end{proof}

If we, in addition, use Lemma \ref{lemma: cubic interaction}
we obtain similar inequalities for the Hamiltonian $H(t)$.
\begin{lemma}
\label{lemma:estimates H(t)}
Let $\varepsilon >0$ be arbitrary. Then, there exists a constant $C_\varepsilon > 1$, depending on $\varepsilon$, such that
\begin{align}
\label{eq:estimates H(t) 1}
\pm i \left[ \mathcal{N} , H(t) \right] 
&\leq \varepsilon \, T_b + C_{\varepsilon} \norm{\psi_t}_{H^1(\mathbb{R}^3)}^2 \left( \mathcal{N} + 1 + N^{-1} \left( \mathcal{N}_a + 1 \right) \mathcal{N}_b \right) ,
\\
\label{eq:estimates H(t) 2}
\pm \left( H(t) - T_b \right)
&\leq \varepsilon T_b + C_{\varepsilon} 
\left( \norm{\psi_t}_{H^1(\mathbb{R}^3)}^2 + \norm{\varphi_t}_{L^2(\mathbb{R}^3)}^2 \right)
\left( \mathcal{N} + 1 + N^{-1} \left( \mathcal{N}_a + 1 \right) \mathcal{N}_b \right) ,
\\
\label{eq:estimates H(t) 3}
\pm \frac{d}{dt} H(t)
&\leq \varepsilon T_b + C_{\varepsilon} \left( \norm{\psi_t}_{H^3(\mathbb{R}^3)}^2 + \norm{\psi_t}_{H^1(\mathbb{R}^3)}^2  \norm{\varphi_t}_{L_{2}^2(\mathbb{R}^3)}^2 \right) 
\left( \mathcal{N} + 1 + N^{-1} \left( \mathcal{N}_a + 1 \right) \mathcal{N}_b \right) .
\end{align}
\end{lemma}

\begin{proof}
Using $\left[ 1 - N^{-1} \mathcal{N}_b \right]_{+}^{1/2} \leq 1$ the first three terms of $H(t)$ on the right hand side of \eqref{eq:H on the double Fock space} can be estimated in exactly the same way as in Lemma \ref{lemma:estimates H-0(t)}. We consequently only have to consider the term with three annihilation and creation operators. Since
\begin{align}
&\left[ \mathcal{N} , \int dx\, b^{*}_x \Big( q(t) \widehat{\Phi} q(t) - \scp{\psi_t}{\widehat{\Phi} \, \psi_t}_{L^2(\mathbb{R}^3)} \Big) b_x \right]
\nonumber \\
&\quad =  \int dx \, b_x^* \left( q(t) \left( a^*(G_{\cdot}) - a(G_{\cdot})  \right) q(t) - \scp{\psi_t}{ \left( a^*(G_{\cdot}) - a(G_{\cdot})  \right) \psi_t}  \right) b_x
\end{align}
we obtain \eqref{eq:estimates H(t) 1} by similar estimates as in  the proof of Lemma \ref{lemma: cubic interaction} and \eqref{eq:estimates H-0(t) 1}. Inequality \eqref{eq:estimates H-0(t) 2} in combination with Lemma \ref{lemma: cubic interaction} leads to \eqref{eq:estimates H(t) 2}.
Using $\dot{q}(t;x;y)= - \dot{\psi_t}(x) \overline{\psi_t(y)} - \psi_t(x) \overline{\dot{\psi_t}(y)}$ we compute
\begin{align}
&\frac{d}{dt} \int dx \, b_x^* \left( q(t) \widehat{\Phi} q(t) - \scp{\psi_t}{\widehat{\Phi} \psi_t}_{L^2(\mathbb{R}^3)} \right) b_x
\nonumber \\
\label{eq:cubic term time derivative 1}
&\quad = - \left(
b^* \left( \psi_t \right) \int dy \, \overline{\dot{\psi_t}(y)} \, \widehat{\Phi}(y) \left( G_y \right) b_y
+ b^* (\dot{\psi_t}) \int dy \, \overline{\psi_t(y)}  \widehat{\Phi}(y) b_y 
+ \text{h.c.} \right)
\\
\label{eq:cubic term time derivative 2}
&\qquad + 
b^* ( \dot{\psi_t} ) \sum_{\sharp \in \{ \cdot , * \}} a^{\sharp}\left( \scp{\psi_t}{G_{\cdot} \psi_t}_{L^2(\mathbb{R}^3)} \right) b (\psi_t) + \text{h.c.}
\\
\begin{split}
\label{eq:cubic term time derivative 3}
&\qquad +
\sum_{\sharp \in \{ \cdot , * \}} \left( b^* (\psi_t)  \bigg(
a^{\sharp} \big( \scp{\dot{\psi_t}}{G_{\cdot} \psi_t}_{L^2(\mathbb{R}^3)} \big)
+ a^{\sharp} \big( \scp{\psi_t}{G_{\cdot} \dot{\psi_t}}_{L^2(\mathbb{R}^3)} \big)
\right) b \left( \psi_t \right)
\\
&\qquad \qquad \qquad \qquad 
-  \frac{d}{dt} a^{\sharp}\left( \scp{\psi_t}{G_{\cdot} \psi_t}_{L^2(\mathbb{R}^3)} \right) \mathcal{N}_b  \bigg).
\end{split}
\end{align}
By means of \eqref{eq:commutator with G-x} and integration by parts, i.e. the commutator method of Lieb and Yamazaki \cite{LY1958}, we obtain
\begin{align}
\pm \eqref{eq:cubic term time derivative 1}
&\leq \varepsilon \left( \mathcal{N}_b + T_b \right)
+ C \varepsilon^{-1} 
\left( \| \dot{\psi_t} \|_{H^1(\mathbb{R}^3)}^2 + \| \dot{\psi_t} \|_{L^2(\mathbb{R}^3)}^2 \norm{\psi_t}_{H^1(\mathbb{R}^3)}^2 \right) \left( \mathcal{N}_a + 1 \right) \mathcal{N}_b .
\end{align}
Using \eqref{eq:estimate expectation value of G} we get
\begin{align}
\pm  \eqref{eq:cubic term time derivative 2}
&\leq C  \norm{\psi_t}_{H^1(\mathbb{R}^3)} 
\| \dot{\psi_t} \|_{L^2(\mathbb{R}^3)} 
\left( \mathcal{N}_a + 1 \right)^{1/2} \mathcal{N}_b
\end{align}
and 
\begin{align}
\pm  \eqref{eq:cubic term time derivative 3}
&\leq C \left( \norm{\psi_t}_{H^1(\mathbb{R}^3)}  \| \dot{\psi_t} \|_{L^2(\mathbb{R}^3)}
+  \| \dot{\psi_t} \|_{H^1(\mathbb{R}^3)} \right) \left( \mathcal{N}_a + 1 \right)^{1/2} \mathcal{N}_b .
\end{align}
Together with \eqref{eq:norm bound for the derivative of psi-t} this leads to
\begin{align}
&\pm N^{-1/2} \frac{d}{dt}  \int dx \, b_x^* \left( q(t) \widehat{\Phi} q(t) - \scp{\psi_t}{\widehat{\Phi} \psi_t}_{L^2(\mathbb{R}^3)} \right) b_x
\nonumber \\
&\quad \leq \varepsilon T_b +
 C_{\varepsilon} \left( \norm{\psi_t}_{H^3(\mathbb{R}^3)}^2 +  \norm{\psi_t}_{H^1(\mathbb{R}^3)}^2 \norm{\varphi_t}_{L_2^2(\mathbb{R}^3)}^2 \right) 
\left( \mathcal{N}_b  + N^{-1} \left( \mathcal{N}_a + 1 \right) \mathcal{N}_b \right).
\end{align}
Thus if we combine this estimate with \eqref{eq:estimates H-0(t) 3} we obtain \eqref{eq:estimates H(t) 3}.
\end{proof}

\section{Proofs}
\label{sec:proofs}
In this section we prove the main results of the article. We start with Theorem \ref{theorem:reduced density matrices}. Afterwards, we discuss the well-posedness of the Bogoliubov dynamics, introduce for technical reasons a Bogoliubov evolution which is truncated in the total number of excitations and finally derive Theorem \ref{theorem:Bogoliubov approximation}.
It is convenient to consider solutions $\ch(t)$ of the Schr\"odinger equation \eqref{eq:Schroedinger equations in double excitation space} rather on $\mathcal{G}$ than on the time dependent truncated excitation space $\mathcal G^{\leq N}_{\perp \psi_t}$. We therefore define $\chi(t) \in \mathcal{G}$ by \footnote{Note that we refrain from indicating the dependence of $\chi(t)$ on $N$ to simplify the notation.}
\begin{align}
\label{eq:definition many-body state on full double Fock space}
\chi^{(k)}(t) =
\begin{cases}
\ch^{(k)}(t) \quad &\text{if} \; k \in \{1, 2, \ldots, N \} ,
\\
0 \quad &\text{else} ,
\end{cases}
\end{align}
which satisfies the Schr\"odinger equation
\begin{align}
\label{eq:Schroedinger equation full Fock space}
i \partial_t \chi(t) = H(t) \chi(t) 
\quad \text{with} \quad 
\chi^{(k)}(0) =
\begin{cases}
\left( U_N(0) \Psi_{N,0} \right)^{(k)} \quad &\text{if} \; k \in \{1, 2, \ldots, N \} ,
\\
0 \quad &\text{else} .
\end{cases}
\end{align}

\subsection{Convergence of reduced density matrices}

\begin{proof}[Proof of Theorem \ref{theorem:reduced density matrices}]
Note that 
\begin{align}
\label{eq:estimate Pekar energy}
\norm{\psi_t}_{H^1(\mathbb{R}^3)}^2 + \norm{\varphi_t}_{L^2(\mathbb{R}^3)}^2 \leq 2 \left(  \mathcal{E}[\psi_t, \varphi_t] + C \right) = 2 \left( \mathcal{E}[\psi, \varphi] + C \right) 
&\leq C \big( \norm{\psi}_{H^1(\mathbb{R}^3)}^2 + \norm{\varphi}_{L^2(\mathbb{R}^3)}^2 \big)
\end{align}
holds because of \eqref{eq:estimate expectation value of G} 
and the conservation of energy, see Proposition \eqref{proposition:solution theory for LP equation}. According to Lemma \ref{lemma:estimates H(t)} there exists $\widetilde{C} > 0$ which only depends on $\mathcal{E}[\psi,\varphi]$ such that the operator 
\begin{align}
\label{eq:definition of A(t)}
A(t) =  H(t) 
+ \widetilde{C} \left( \mathcal{N} + 1 \right)
\end{align}
satisfies
\begin{align}
\label{eq:commutator estiamte for truncated many-body Hamiltonian 1}
\id_{\mathcal{N}_b \leq N} \left( \mathcal{N} + T_b + 1 \right)  \id_{\mathcal{N}_b \leq N}
&\leq 2 \id_{\mathcal{N}_b \leq N} A(t) \id_{\mathcal{N}_b \leq N} ,
\\
\label{eq:commutator estiamte for truncated many-body Hamiltonian 2}
\id_{\mathcal{N}_b \leq N} A(t) \id_{\mathcal{N}_b \leq N}
&\leq  3 \widetilde{C}
\id_{\mathcal{N}_b \leq N}
\left( \mathcal{N} + T_b + 1 \right) 
\id_{\mathcal{N}_b \leq N} ,
\\
\label{eq:commutator estiamte for truncated many-body Hamiltonian 3}
\pm  \id_{\mathcal{N}_b \leq N} i \left[ H(t) , A(t) \right] \id_{\mathcal{N}_b \leq N} 
&\leq   \widetilde{C}
 \id_{\mathcal{N}_b \leq N} A(t) \id_{\mathcal{N}_b \leq N} , 
\\
\label{eq:commutator estiamte for truncated many-body Hamiltonian 4}
\pm  \id_{\mathcal{N}_b \leq N} \frac{d}{dt} A(t) \id_{\mathcal{N}_b \leq N}
&\leq  
\left( \norm{\psi_t}_{H^3(\mathbb{R}^3)}^2 + \norm{\varphi_t}_{L_2^2(\mathbb{R}^3)}^2 \right) \id_{\mathcal{N}_b \leq N} A(t) \id_{\mathcal{N}_b \leq N} .
\end{align}
Let $\Psi_{N,t}$ be the solution of the Schr\"odinger equation of Theorem \ref{theorem:reduced density matrices} and $\chi(t)$ be defined as in \eqref{eq:definition many-body state on full double Fock space}. By means of \eqref{eq:Schroedinger equation full Fock space}, \eqref{eq:commutator estiamte for truncated many-body Hamiltonian 3}, \eqref{eq:commutator estiamte for truncated many-body Hamiltonian 4} and  $\chi(t) = \id_{\mathcal{N}_b \leq N} \chi(t)$ we estimate
\begin{align}
\label{eq:derivation reduced densities time derivative 1}
\abs{\frac{d}{dt} \scp{\chi(t)}{A(t) \chi(t)} }
&\leq  \abs{\scp{\chi(t)}{\dot{A}(t) \chi(t)}}
+ \abs{ \scp{\chi(t)}{\left[ H(t) , A(t) \right] \chi(t)} }
\\
&\leq \widetilde{C} \left( \norm{\psi_t}_{H^3(\mathbb{R}^3)}^2 + \norm{\varphi_t}_{L_2^2(\mathbb{R}^3)}^2  \right) \scp{\chi(t)}{A(t) \chi(t)} .
\end{align}
Using Gronwall's lemma we get
\begin{align}
\scp{\chi(t)}{A(t) \chi(t)} 
&\leq e^{\widetilde{C} \int_0^t ds \,
\left( \norm{\psi_s}_{H^3(\mathbb{R}^3)}^2 + \norm{\varphi_s}_{L_2^2(\mathbb{R}^3)}^2 \right) } \scp{\chi(0)}{A(0) \chi(0)} .
\end{align}
Inequalities \eqref{eq:commutator estiamte for truncated many-body Hamiltonian 1} and \eqref{eq:commutator estiamte for truncated many-body Hamiltonian 2} then lead to
\begin{align}
\label{eq:estimate number operator many-body dynamics}
\scp{\chi(t)}{\left( \mathcal{N} + T_b + 1 \right) \chi(t)}
&\leq 6 \widetilde{C}  e^{\widetilde{C} \int_0^t ds \,
\left( \norm{\psi_s}_{H^3(\mathbb{R}^3)}^2 + \norm{\varphi_s}_{L_2^2(\mathbb{R}^3)}^2 \right) } 
\scp{\chi(0)}{\left( \mathcal{N} + T_b + 1 \right) \chi(0)} .
\end{align}
Note that 
\begin{align}
\label{eq:expectation of number operators relation between truncated and non truncated many-body state}
\scp{\chi(t)}{\left( \mathcal{N} + T_b + 1 \right) \chi(t)}_{\mathcal{G}} 
& = \scp{\ch(t)}{\left( \mathcal{N}_a + \mathcal{N}_b + T_b + 1 \right) \ch(t)}_{\mathcal G^{\leq N}_{\perp \psi_t}}
\end{align}
and that the unitary mapping \eqref{eq:action of the unitary} can be written\footnote{We refer to \cite[Lemma A.1]{FLMP2021}) for a thorough introduction to $U_N(t)$ and its properties.} as $U_{N}(t) = \widetilde{U}_N(t) \otimes W^* \left( \sqrt{N} \varphi_t \right)$ where $\widetilde{U}_N(t)$ is the excitation map from \cite[Chapter 2.5]{LNSS2015}. Since $\left[ \widetilde{U}_N(t) \otimes \id_{\mathcal{F}} , \id_{\left(L^2 \left( \mathbb{R}^{3} \right) \right)^{\otimes_s N}} \otimes W^*(\sqrt{N} \varphi_t) \right] = 0$ we have
\begin{align}
\label{eq:action of excitation map on the number of phonon operator}
\scp{\ch(t)}{\mathcal{N}_a \ch(t)}_{\mathcal G^{\leq N}_{\perp \psi_t}}
&= \scp{\Psi_{N,t}}{W ( \sqrt{N} \varphi_t ) \mathcal{N}_a W^* ( \sqrt{N} \varphi_t ) \Psi_{N,t}}_{\mathcal{H}^{(N)}} 
\end{align}
and 
\begin{align}
\label{eq:action of excitation map on b-star and b}
U_N(t) b^* (f) b(g) U_N(t)^*
&=  b^*(f) b(g) \quad \text{for all} \;
f,g \in L^2_{\perp \psi_t}(\mathbb{R}^3) 
\end{align}
by means of \cite[Proposition 4.2]{LNSS2015}.  For the reduced density $\gamma^{(1,0)}_{\ch(t)}$ with integral kernel
\begin{align}
\gamma^{(1,0)}_{\ch(t)}(x;y) = N^{-1} \scp{\ch(t)}{b_x^* b_y \ch(t)}
\end{align}
relation \eqref{eq:action of excitation map on b-star and b} and $b (\psi_t) \ch(t) = 0$ lead to
\begin{align}
\gamma^{(1,0)}_{\ch(t)} = q(t) \gamma^{(1,0)}_{\Psi_{N,t}} q(t)
\end{align}
and
\begin{align}
\label{eq:action of excitation map on the number and kinetic energy of the particles}
N^{-1} \scp{\ch(t)}{\left( \mathcal{N}_b + T_b \right) \ch(t)} = \tr \left( (1 - \Delta)  \gamma^{(1,0)}_{\ch(t)}  \right) 
= \tr \left( (1 - \Delta) q(t) \gamma^{(1,0)}_{\Psi_{N,t}} q(t) \right) .
\end{align}
Using
\begin{align}
&\tr \left|  \sqrt{1 - \Delta} \left( \gamma^{(1,0)}_{\Psi_{N,t}} - \ket{\psi_t} \bra{\psi_t} \right)  \sqrt{1 - \Delta} \right|
\nonumber \\
&\quad = \sup_{A \in \mathcal{L}^{\infty}(L^2(\mathbb{R}^3), L^2(\mathbb{R}^3)), \norm{A}_{\mathcal{L}^{\infty}} = 1} \abs{ \tr \left(  \sqrt{1 - \Delta}  A \sqrt{1 - \Delta} \left( \gamma^{(1,0)}_{\Psi_{N,t}} - \ket{\psi_t} \bra{\psi_t} \right)  \right) } ,
\end{align}
the identity $\id_{L^2(\mathbb{R}^3)} = q(t) + \ket{\psi_t} \bra{\psi_t}$ and the Cauchy--Schwarz inequality we obtain, in analogy to \cite[Lemma VII.1]{LP2018} (see also \cite[Proof of Theorem 2.8]{MPP2019}),
\begin{align}
\label{eq:inequality Sobolev trace norm}
\tr \left|  \sqrt{1 - \Delta} \left( \gamma^{(1,0)}_{\Psi_{N,t}} - \ket{\psi_t} \bra{\psi_t} \right)  \sqrt{1 - \Delta} \right|
&\leq C \norm{\psi_t}_{H^1(\mathbb{R}^3)} \sup_{j=1,2}
\left(\tr \left( (1 - \Delta) q(t) \gamma^{(1,0)}_{\Psi_{N,t}} q(t) \right) \right)^{j/2}.
\end{align}
Combining \eqref{eq:estimate number operator many-body dynamics} with \eqref{eq:expectation of number operators relation between truncated and non truncated many-body state},  \eqref{eq:action of excitation map on the number of phonon operator}, \eqref{eq:action of excitation map on the number and kinetic energy of the particles} and \eqref{eq:inequality Sobolev trace norm} 
proves \eqref{eq: theorem reduced density matrices estimate 1} and 
\eqref{eq: theorem reduced density matrices estimate 2}.
\end{proof}

\begin{remark}
The derivation of \eqref{eq:estimate number operator many-body dynamics} from above is in our opinion the most insightful but we would like to remark that it is rather formal because the second term on the right hand side of \eqref{eq:derivation reduced densities time derivative 1} is not well defined for all $\chi(t) \in \mathcal{D} \big( \big( \sum_{j=1}^N - \Delta_j + \mathcal{N}_a \big)^{1/2} \big)$.
A rigorous derivation is obtained if one proceeds in analogy to the proof of \cite[Theorem 8]{LNS2015} and considers a regularized version of \eqref{eq:Schroedinger equation full Fock space}. Likewise one can replace $H(t)$ in \eqref{eq:Schroedinger equation full Fock space} by $\id_{\mathcal{N}_b \leq N} H(t) \id_{\mathcal{N}_b \leq N}$ and directly apply \cite[Theorem 8]{LNS2015} with $A = \id_{\mathcal{N}_b \leq N} \left( \mathcal{N}+ T_b \right) \id_{\mathcal{N}_b \leq N} + 1$ and $B = \id_{\mathcal{N}_b \leq N}  \mathcal{N} \id_{\mathcal{N}_b \leq N} + 1$.
\end{remark}

\subsection{Well-posedness of the Bogoliubov dynamics}

Next, we are going to show that \eqref{eq:norm approximation main result approximating many-body state} approximates the time evolved many-body state in norm. We start with commenting on the well-posedness of the Bogoliubov dynamics. Afterwards, we will introduce a truncated Bogoliubov dynamics which will be used in the proof of Theorem \ref{theorem:Bogoliubov approximation}. 

\begin{lemma}
\label{lemma:well-posedness Bogoliubov dynamics}
Let $(\psi, \varphi) \in H^3(\mathbb{R}^3) \times L_2^2(\mathbb{R}^3)$ such that $\norm{\psi}_{L^2(\mathbb{R}^3)} = 1$ and $(\psi_t, \varphi_t)$ be the unique solution of \eqref{eq:Landau Pekar equations} with initial datum $(\psi, \varphi)$. For every $\chi \in \mathcal{G} \cap Q \left( \mathcal{N} + T_b \right)$ there exists a unique solution to the Bogoliubov equation \eqref{eq:Bogoliubov dynamics}
with $\chb(0) = \chi$ such that $\chb \in C^0 \left( [0,\infty) \cap \mathcal{G} \right) \cap L_{\rm loc}^{\infty} \left( [0, \infty) ; Q \left( \mathcal{N} + T_b \right) \right)$. Moreover, there exists a constant $C >0$ depending only on $\mathcal{E}[\psi,\varphi]$ such that
\begin{align}
\label{eq:Bog dynamics estimate number operators}
\scp{\chb(t)}{\left( \mathcal{N} + T_b + 1 \right) \chb(t)}
&\leq C  e^{C \int_0^t ds \,
\left( \norm{\psi_s}_{H^3(\mathbb{R}^3)}^2 + \norm{\varphi_s}_{L_2^2(\mathbb{R}^3)}^2 \right) } 
\scp{\chi}{\left( \mathcal{N} + T_b + 1 \right) \chi}
\end{align}
and the condition $\chi \in \mathcal G_{\perp \psi_0} $ implies $\chb(t) \in \mathcal G_{\perp \psi_t}$ for all $t \in \mathbb{R}$.
\end{lemma}

\begin{proof}
Let $A= \mathcal{N} + T_b + 1$ and $B = \mathcal{N} + 1$.
By Lemma \ref{lemma:estimates H-0(t)} and \eqref{eq:estimate Pekar energy} there exists a constant $C > 0$ depending only on $\mathcal{E}[\psi, \varphi]$ such that 
\begin{align}
\begin{split}
\label{eq:Bounds for the well-posedness of Bogoliubov}
C^{-1} A - C B &\leq H^{\rm B}(t) \leq C A ,
\\
\pm i \left[ H^{\rm B}(t), B \right] &\leq C A ,
\\
\rm \frac{d}{dt} H^{\rm B}(t)
&\leq C \left( \norm{\psi_t}_{H^3(\mathbb{R}^3)}^2 + \norm{\varphi_t}_{L_2^2(\mathbb{R}^3)}^2 \right) A .
\end{split}
\end{align}
The statement of Lemma \ref{lemma:well-posedness Bogoliubov dynamics} until \eqref{eq:Bog dynamics estimate number operators} then follows from \cite[Theorem 8]{LNS2015}. Note that  the time dependence of \eqref{eq:Bounds for the well-posedness of Bogoliubov} must be tracked in the proof of \cite[Theorem 8]{LNS2015} to obtain the explicit form of the exponent in \eqref{eq:Bog dynamics estimate number operators}.
Let us define $\Gamma_t: \mathcal{G} \rightarrow \mathcal G_{\perp \psi_t}$ by
$\Gamma_t |_{\mathcal{F}_b^{(j)} \otimes \mathcal{F}_a} = q(t)^{\otimes j} \otimes \id_{\mathcal{F}_a}$ and compute 
\begin{align}
\frac{d}{dt} \norm{\Gamma_t \chb(t)}^2
&= \scp{\chb(t)}{i \left[ H^{\rm B}(t), \Gamma_t \right] \chb(t)}
+ \scp{\chb(t)}{\dot{\Gamma}_t H^{\rm B}(t)} = 0 .
\end{align}
Here, we have used that the relations
\begin{align}
\dot{\Gamma}_t
&= - b^*(\psi_t) b (q(t) \dot{\psi}_t) \Gamma_t - \Gamma_t b^*(q(t) \dot{\psi}_t ) b(\psi_t)
\end{align}
and 
\begin{align}
i \left[ H^{\rm B}(t) , \Gamma_t \right]
&= b^* (\psi_t) b (q(t) \dot{\psi}_t)
\Gamma_t + \Gamma_t b( q(t) \dot{\psi}_t) b(\psi_t)
\end{align}
can obtained (in analogy to \cite[p. 1588]{BS2019}) by a direct calculation on the Fock space sector with $k$ particles. This leads to
$\norm{\left( 1 - \Gamma_t \right) \chb(t)}^2
= \norm{\left( 1 - \Gamma_0 \right) \chb(0)}^2$
and shows that $\chi \in \mathcal G_{\perp \psi_0} $ implies $\chb(t) \in \mathcal G_{\perp \psi_t}$ for all $t \in \mathbb{R}$.
\end{proof}

As a technical tool we introduce (for $M \in \mathbb{N}$ arbitrary but fixed) the truncated Bogoliubov dynamics
\begin{align}
\label{eq:truncated Bogoliubov dynamics}
i \partial_t \chbm(t) = \id_{\mathcal{N} \leq M} H^{\rm{B}}(t) \id_{\mathcal{N} \leq M} \chbm(t) ,
\quad \chbm(0) =  \id_{\mathcal{N} \leq M}  \chb(0) .
\end{align}
In the following, we use the shorthand notations $\id_{\leq M} = \id_{\mathcal{N} \leq M} $ and
$\id_{> M} = \id_{\mathcal{N} >M} $. The truncated dynamics satisfies the same existence result as the original Bogoliubov equation.
\begin{lemma}
\label{lemma:well-posedness truncated Bogoliubov dynamics}
Let $(\psi, \varphi) \in H^3(\mathbb{R}^3) \times L_2^2(\mathbb{R}^3)$ such that $\norm{\psi}_{L^2(\mathbb{R}^3)} = 1$ and $(\psi_t, \varphi_t)$ be the unique solution of \eqref{eq:Landau Pekar equations} with initial datum $(\psi, \varphi)$. For every $\chi \in \mathcal{G} \cap Q \left( \mathcal{N} + T_b \right)$ there exists a unique solution to the truncated Bogoliubov equation \eqref{eq:truncated Bogoliubov dynamics}
with $\chb(0) = \chi$ such that $\chbm \in C^0 \left( [0,\infty) \cap \mathcal{G} \right) \cap L_{\rm loc}^{\infty} \left( [0, \infty) ; Q \left( \mathcal{N} + T_b \right) \right)$. Furthermore, $\id_{\mathcal{N} >M} \chbm(t) = 0$ holds for all $t \in \mathbb{R}$ and $\chi \in \mathcal{G}_{\perp \psi_0}$ implies $\chbm(t) \in \mathcal{G}_{\perp \psi_t}$ for all $t \in \mathbb{R}$.

In addition, assume the existence of a constant $\widetilde{C} >0$ such that
$
\big\| ( \Number^{3/2} + T_b^{1/2} + 1 ) \chi \big\| \leq \widetilde{C}
$.
Then, there exists a constant $C>0$ (depending only on $\widetilde{C}$ and $\mathcal{E}[\psi,\varphi]$)  such that
\begin{align}
\label{eq:propgation of moments truncated Bogoliubov dynamics}
\sup \left\{ \norm{\left( \mathcal{N} + T_b + 1 \right)^{1/2} \chbm(t)},
M^{-1/4} \norm{\left( \mathcal{N} + 1 \right) \chbm(t)},
M^{-5/8} \norm{\left( \mathcal{N} + 1 \right)^{3/2} \chbm(t)}
\right\}
&\leq C e^{C f(t)}  
\end{align}
with $f(t) = \int_0^t ds \, \left( \norm{\psi_s}_{H^3(\mathbb{R}^3)}^2 + \norm{\varphi_s}_{L_2^2(\mathbb{R}^3)}^2 \right)$.
\end{lemma} 
\begin{proof}
Note that $\id_{\leq M} H^{\rm{B}}(t) \id_{\leq M}$ satisfies the same estimates \eqref{eq:Bounds for the well-posedness of Bogoliubov} as $H^{\rm{B}}(t)$ if one replaces $A$ and $B$ by $\id_{\leq M} (\mathcal{N} + T_b) \id_{\leq M} + 1$ and $\id_{\leq M} \mathcal{N} \id_{\leq M} + 1$. By \cite[Theorem 8]{LNS2015} it follows that for every $\chi \in \mathcal{G} \cap Q \left( \mathcal{N} + T_b \right)$ there exists a unique solution to the truncated Bogoliubov equation \eqref{eq:truncated Bogoliubov dynamics}
with $\chb(0) = \chi$ such that $\chbm \in C^0 \left( [0,\infty) \cap \mathcal{G} \right) \cap L_{\rm loc}^{\infty} \left( [0, \infty) ; Q \left( \id_{\leq M} \left(\mathcal{N} + T_b \right) \right) \id_{\leq M} \right)$ and $\| \left( \mathcal{N} + T_b \right)^{1/2} \id_{\leq}  \chb(t) \|^2 \leq C e^{C f(t)}$.  Since $\frac{d}{dt} \norm{\id_{\mathcal{N} > M} \chbm(t)}^2 =0$ we get $\id_{\mathcal{N} >M} \chbm(t) = 0$  for all $t \in \mathbb{R}$. Together with the previous inequality this implies $\| \left( \mathcal{N} + T_b + 1 \right)^{1/2} \chbm(t) \| \leq C e^{C f(t)}$. Using $\left[ \id_{\leq M}, \Gamma_t \right]$ we conclude by similar means as in the proof of Lemma \ref{lemma:well-posedness Bogoliubov dynamics} that $\chbm(t) \in \mathcal{G}_{\perp \psi_t}$ if $\chi \in \mathcal{G}_{\perp \psi_0}$. In total, this proves the first part of the lemma.
Below we will prove
\begin{align}
\label{eq:Gronwall estimate for moments of the number operator}
\norm{\left( \mathcal{N} + 1 \right)^{\frac{k}{2}} \chbm(t)}^2
&\leq e^{C \int_0^t ds \, \norm{\psi_s}_{H^1(\mathbb{R}^3)}^2}
\Bigg[ \norm{\left( \mathcal{N} + 1 \right)^{\frac{k}{2}} \chbm(0)}^2
\nonumber \\
&\quad 
+ \int_0^t ds \, \left( \Lambda \norm{\left( \mathcal{N} + 1 \right)^{\frac{k-1}{2}} \chbm(s)}^2
+ \Lambda^{-1} M^{k-1}  \norm{\left( \mathcal{N} + T_b + 1 \right)^{\frac{1}{2}} \chbm(s)}^2
\right) \Bigg] 
\end{align}
for $k \in \mathbb{N}$ satisfying $k \geq 2$ and $\Lambda \geq 1$.
Choosing $\Lambda = M^{\frac{1}{2}}$ for $k= 2$ and $\Lambda = M^{\frac{3}{4}}$ for $k=3$ proves \eqref{eq:propgation of moments truncated Bogoliubov dynamics}.
It remains to show \eqref{eq:Gronwall estimate for moments of the number operator}. With this regard note that $\left[ \mathcal{N} , \id_{\mathcal{N} \leq M} \right] = 0$ and that $\mathcal{N}$ is a bounded operator on the subspace
$\left\{ \chi \in \mathcal{F} 
\otimes \mathcal{F} :  \id_{\mathcal{N} > M} \chi = 0 \right\}$. Using the shifting properties of $\mathcal{N} = \mathcal{N}_a + \mathcal{N}_b$ we calculate
\begin{align}
&\frac{d}{dt} \norm{\left( \mathcal{N} + 1 \right)^{\frac{k}{2}} \chbm(t)}^2
\nonumber \\
&\quad = i \scp{\chbm(t)}{\left[ H^{\rm{B}}(t) , \left(\mathcal{N} + 1 \right)^k \right] \chbm(t)}
\nonumber \\
\label{eq:Gronwall estimate for moments of the number operator term 1}
&\quad = 2 \Im 
\scp{\chbm(t)}{\left( \left( \mathcal{N} + 3 \right)^k - \left( \mathcal{N} + 1 \right)^k \right) \int dx \, \int dk \, \overline{\psi_t(x)} \, \overline{\scp{\psi_t}{G_{\cdot}(k) \psi_t}}_{L^2(\mathbb{R}^3)}  a_k b_x \chbm(t)}
\\
\label{eq:Gronwall estimate for moments of the number operator term 2}
&\qquad - 2 \Im
\scp{\chbm(t)}{\left( \left( \mathcal{N} + 3 \right)^k - \left( \mathcal{N} + 1 \right)^k \right) \int dx \, \int dk \, \overline{\psi_t(x)} \,  \overline{G_x(k)} a_k b_x \chbm(t)} .
\end{align}
Using again the shifting property of the number operator, \eqref{eq:estimate expectation value of G} and the Cauchy--Schwarz inequality we bound the first term by
\begin{align}
&\abs{\eqref{eq:Gronwall estimate for moments of the number operator term 1}}
\nonumber \\
&\quad = 2 \abs{\scp{\left( \left( \mathcal{N} + 3 \right)^k - \left( \mathcal{N} + 1 \right)^k \right)  \left( \mathcal{N} + 3 \right)^{1 - \frac{k}{2}}\chbm(t)}{ b(\psi_t) a \left( \scp{\psi_t}{G_{\cdot} \psi_t}_{L^2(\mathbb{R}^3)} \right) \left(\mathcal{N} + 1 \right)^{\frac{k}{2} - 1} \chbm(t)}}
\nonumber \\
&\quad \leq 2 \left( \int dk \, \abs{\scp{\psi_t}{G_{\cdot}(k) \psi_t}_{L^2(\mathbb{R}^3)}}^2 \right)^{\frac{1}{2}} \norm{\left( \left( \mathcal{N} + 3 \right)^k - \left(\mathcal{N}+ 1 \right)^k \right)  \left( \mathcal{N} + 3 \right)^{1 - \frac{k}{2}}\chbm(t)}
\norm{\mathcal{N}^{\frac{k}{2}} \chbm(t)}
\nonumber \\
&\leq C \norm{\psi_t}_{H^1(\mathbb{R}^3)} \norm{\left( \mathcal{N} + 1 \right)^{\frac{k}{2}} \chbm(t)}^2 .
\end{align}
To obtain the ultimate inequality we have, in addition, used that 
\begin{align}
\label{eq: estimate for shifted moments of the number of particle operator}
\norm{\left( \left( \mathcal{N} + 3 \right)^k - \left( \mathcal{N} + 1 \right)^k \right) \chbm(t)}
&\leq 2 k \norm{\left( \mathcal{N} + 3 \right)^{k-1} \chbm(t)}
\end{align}
holds by the spectral theorem because $\abs{y^k - x^k} \leq k y^{k-1} (y-x)$ for all $y \geq x \geq 0$.
Next, we write
\eqref{eq:Gronwall estimate for moments of the number operator term 2} as 
\begin{align}
\begin{split}
\label{eq:Gronwall estimate for moments of the number operator term 2 a}
\eqref{eq:Gronwall estimate for moments of the number operator term 2}
&= - 2 \Im
\big< \left( \left( \mathcal{N} + 3 \right)^k - \left( \mathcal{N} + 1 \right)^k \right)  \left( \mathcal{N} + 3 \right)^{\frac{1 - k}{2}} \chbm(t) , \int dx \,
\int_{\abs{k} \leq \Lambda} dk \, \overline{\psi_t(x)} \,  \overline{G_x(k)}
\\
&\qquad \qquad \quad \times
 a_k b_x \left(\mathcal{N} + 1 \right)^{\frac{k-1}{2}}  \chbm(t) \big>
\end{split}
\\
\label{eq:Gronwall estimate for moments of the number operator term 2 b}
&\quad 
- 2 \Im
\scp{\chbm(t)}{\left( \left( \mathcal{N} + 3 \right)^k - \left( \mathcal{N} + 1 \right)^k \right) \int dx \, \int_{\abs{k} \geq \Lambda} dk \, \overline{\psi_t(x)} \, \overline{G_x(k)} a_k b_x \chbm(t)}.
\end{align}
Using the second inequality of Lemma \ref{lemma:preliminary estimates} (note that the two summands on the left hand side of \eqref{eq:preliminary estimates quadratic term 2} are estimated separately) 
let us bound the first summand by
\begin{align}
\abs{\eqref{eq:Gronwall estimate for moments of the number operator term 2 a}}
&\leq C \norm{\left( \mathcal{N}+ 1 \right)^{\frac{k}{2}} \chbm(t)}
\bigg( \Lambda^{1/2} \norm{\left( \left( \mathcal{N} + 3 \right)^k - \left(\mathcal{N} + 1 \right)^k \right)  \left( \mathcal{N} + 3 \right)^{\frac{1 - k}{2}} \chbm(t)}
\nonumber \\
&\qquad \qquad \qquad  \qquad 
+ \norm{\psi_t}_{H^1(\mathbb{R}^3)}
\norm{\left( \left( \mathcal{N} + 3 \right)^k - \left(\mathcal{N} + 1 \right)^k \right)  \left( \mathcal{N} + 3 \right)^{1 - \frac{k}{2}} \chbm(t)}
\bigg)
\nonumber \\
&\leq C \norm{\left(\mathcal{N}+ 1 \right)^{\frac{k}{2}} \chbm(t)}
\bigg( \Lambda^{1/2} \norm{\left( \mathcal{N} + 1 \right)^{\frac{k-1}{2}} \chbm(t)}
+ \norm{\psi_t}_{H^1(\mathbb{R}^3)}
\norm{\left( \mathcal{N} + 1 \right)^{\frac{k}{2}} \chbm(t)}
\bigg) 
\nonumber \\
&\leq C \norm{\psi_t}_{H^1(\mathbb{R}^3)} 
\norm{\left( \mathcal{N} + 1 \right)^{\frac{k}{2}} \chbm(t)}^2
+ \Lambda  \norm{\left( \mathcal{N} + 1 \right)^{\frac{k-1}{2}} \chbm(t)}^2 .
\end{align}
By the first inequality of Lemma \ref{lemma:preliminary estimates} and \eqref{eq: estimate for shifted moments of the number of particle operator} we obtain
\begin{align}
\abs{\eqref{eq:Gronwall estimate for moments of the number operator term 2 b}}
&\leq C \norm{\psi_t}_{H^1(\mathbb{R}^3)} \Lambda^{- \frac{1}{2} } 
\norm{\left( \mathcal{N} + T_b \right)^{\frac{1}{2}} \chbm(t)}
\norm{\left( \mathcal{N}_a + 1 \right)^{\frac{1}{2}}\left( \left( \mathcal{N} + 3 \right)^k - \left( \mathcal{N} + 1 \right)^k \right) \chbm(t)}
\nonumber \\
&\leq C \norm{\psi_t}_{H^1(\mathbb{R}^3)} \Lambda^{- \frac{1}{2} }
\norm{\left( \mathcal{N} + T_b \right)^{\frac{1}{2}} \chbm(t)}
\norm{\left( \mathcal{N} + 1 \right)^{k - \frac{1}{2}} \chbm(t)}
\nonumber \\
&\leq C \norm{\psi_t}_{H^1(\mathbb{R}^3)} \Lambda^{- \frac{1}{2} } M^{\frac{k-1}{2}}
\norm{\left( \mathcal{N} + T_b \right)^{\frac{1}{2}} \chbm(t)}
\norm{\left( \mathcal{N} + 1 \right)^{\frac{k}{2}} \chbm(t)}
\nonumber \\
&\leq C \norm{\psi_t}_{H^1(\mathbb{R}^3)}^2
\norm{\left( \mathcal{N} + 1 \right)^{\frac{k}{2}} \chbm(t)}^2 
+ 
\Lambda^{- 1 } M^{k-1}
\norm{\left( \mathcal{N} + T_b \right)^{\frac{1}{2}} \chbm(t)}^2 .
\end{align}
In total, we get
\begin{align}
\frac{d}{dt} \norm{\left( \mathcal{N} + 1 \right)^{\frac{k}{2}} \chbm(t)}^2
&\leq C \norm{\psi_t}_{H^1(\mathbb{R}^3)}^2 
\norm{\left( \mathcal{N} + 1 \right)^{\frac{k}{2}} \chbm(t)}^2
+ \Lambda  \norm{\left( \mathcal{N} + 1 \right)^{\frac{k-1}{2}} \chbm(t)}^2 
\nonumber \\
&\quad + 
\Lambda^{- 1 } M^{k-1}
\norm{\left( \mathcal{N} + T_b \right)^{\frac{1}{2}} \chbm(t)}^2 .
\end{align}
Inequality \eqref{eq:Gronwall estimate for moments of the number operator} then follows by Gronwall's lemma.
\end{proof}

The following Lemma compares the Bogoliubov dynamics to the one with cutoff in the total number of particles.

\begin{lemma}
\label{lemma:Bog dynamics with and without cutoff norm estimate}
Let $M \in \mathbb{N}$ such that $M \geq 3$, $\widetilde{C} >0$ and $\chi \in  \mathcal{G}$ such that $\big\| ( \Number^{3/2} + T_b^{1/2} + 1 ) \chi \big\| \leq \widetilde{C}$. Let $\chb(t)$ and $\chbm(t)$ be the unique solutions of \eqref{eq:Bogoliubov dynamics} and \eqref{eq:truncated Bogoliubov dynamics} with $\chb(0) = \chi$. Then, there exists a constant $C>0$ (depending only on $\widetilde{C}$ and $\mathcal{E}[\psi,\varphi]$)  such that
\begin{align}
\label{eq:Bog dynamics with and without cutoff norm estimate}
\norm{\chb(t) - \chbm(t)}^2
&\leq  C e^{C f(t)}  M^{-3/8}.
\end{align} 
\end{lemma}

\begin{proof}
We have
\begin{align}
\frac{d}{dt}\norm{\chb(t) - \chbm(t)}^2
&= 2 \Im \scp{\chb(t)}{\left( H^{\rm{B}}(t) - \id_{\leq M} H^{\rm{B}}(t) \id_{\leq M} \right) \chbm(t)}
\nonumber \\
&= 2 \Im \scp{\chb(t)}{ \id_{> M} \int dx \, \int dk \,  K(t,k,x) a_k^*  b_x^*   \id_{\leq M}  \chbm(t)}
\nonumber \\
&= 2 \Im \scp{\chb(t)}{ \id_{\mathcal{N} > M} \int dx \, \int dk \,  K(t,k,x) a_k^*  b_x^*   \id_{M - 2 \leq \mathcal{N} \leq M}  \chbm(t)}
\end{align}
because $\chi_{B,M}(t) = \id_{\leq M} \chi_{B,M}(t)$ and other contributions of $H^{\rm{B}}(t)$ map $\left\{ \chi \in \mathcal{G}:  \id_{\mathcal{N} > M} \chi = 0 \right\}$ into itself. Note that
\begin{align}
\abs{\scp{\chi}{\int dx  \int dk \, K(t,k,x) a_k^* b_x^* \Psi}} 
&\leq C \norm{\psi_t}_{H^1} \norm{\left( T_b + \mathcal{N}_b \right)^{1/2} \chi} \norm{\left( \mathcal{N}_a + 1 \right)^{1/2} \Psi}
\end{align}
can be shown in complete analogy to \eqref{eq:preliminary estimates quadratic term 3}. Together with $\left[\left( T_b + \mathcal{N}_b \right), \id_{\mathcal{N} > M} \right]$, Lemma \ref{lemma:well-posedness Bogoliubov dynamics} and Lemma \ref{lemma:well-posedness truncated Bogoliubov dynamics} we get
\begin{align}
\frac{d}{dt}\norm{\chb(t) - \chbm(t)}^2 
&\leq C  \norm{\psi_t}_{H^1}
\norm{\left( T_b + \mathcal{N}_b \right)^{1/2} \id_{\mathcal{N} > M} \chb(t)}
\norm{\left( \mathcal{N}_a + 1 \right)^{1/2} \id_{M - 2 \leq \mathcal{N} \leq M}  \chbm(t)}
\nonumber \\
&\leq C  \norm{\psi_t}_{H^1}
\norm{\left( T_b + \mathcal{N}_b \right)^{1/2} \chb(t)}
\norm{\left( \mathcal{N}_a + 1 \right)^{1/2} \id_{M - 2 \leq \mathcal{N} \leq M}  \chbm(t)}
\nonumber \\
&\leq C (M-2)^{-1}  \norm{\psi_t}_{H^1}
\norm{\left( T_b + \mathcal{N}_b \right)^{1/2} \chb(t)}
\norm{\left( \mathcal{N} + 1 \right)^{3/2}   \chbm(t)} 
\nonumber \\
&\leq C e^{C f(t)}  M^{- 3/8} .
\end{align} 
Here, $f(t) = \int_0^t ds \, \left( \norm{\psi_s}_{H^3(\mathbb{R}^3)}^2 + \norm{\varphi_s}_{L_2^2(\mathbb{R}^3)}^2 \right)$ and $C$ depends only on $\widetilde{C}$ and $\mathcal{E}[\psi,\varphi]$. Using 
\begin{align}
\norm{\chb(0) - \chbm(0)}^2 
&= \norm{\id_{\mathcal{N} > M} \chi}^2
\leq M^{-3} \norm{ \id_{\mathcal{N} > M} \mathcal{N}^{3/2} \chi} \leq \widetilde{C} M^{-3}
\end{align}
and Duhamel's formula shows the claim. 
\end{proof}

\subsection{Norm approximation}

\begin{proof}[Proof of Theorem \ref{theorem:Bogoliubov approximation}]

Since $\chi \in \mathcal{G}_{\perp \psi_0}$ we have $\chi_{\le N}(t), \big( \chb^{(k)}(t) \big)_{k=0}^N  \in \mathcal G^{\leq N}_{\perp \psi_t}$ for all $t \in \mathbb{R}$. 
Using \eqref{eq:decomposition many-body state} and \eqref{eq:norm approximation main result approximating many-body state} we estimate
\begin{align}
\label{eq:norm-approximation first step in the proof}
\norm{\Psi_{N,t} - \Psi_{N,t}^{\rm{B}}}_{\mathcal{H}^{(N)}}
&= 
\norm{\left( \ch^{(k)}(t) - \chb^{(k)}(t) \right)_{k=0}^N}_{\mathcal{G}_{\leq N}}
\nonumber \\
&\leq \norm{\chi(t) - \chb(t)}_{\mathcal{G}}
\nonumber \\
&\leq 
\norm{\chi(t) - \chbm(t)} 
+  \norm{\chb(t) - \chbm(t)}
\end{align}
with $\chi(t)$ and $\chbm(t)$ being defined as in \eqref{eq:definition many-body state on full double Fock space} and \eqref{eq:truncated Bogoliubov dynamics}. Because of \eqref{eq:Schroedinger equation full Fock space} and $\chbm(t) = \id_{\leq M} \chbm(t)$ we get
\begin{align}
\frac{d}{dt}\norm{\chi(t) - \chbm(t)}^2
&= 2 \Im \scp{\chi(t)}{\big( H(t) - \id_{\leq M} H^{\rm{B}}(t) \id_{\leq M} \big) \chbm(t)}
\nonumber \\
&= 2 \Im \scp{\chi(t)}{\big( H(t) - \id_{\leq M} H^{\rm{B}}(t) \big) \id_{\leq M}  \chbm(t)} .
\end{align}
Note that 
\begin{subequations}
\begin{align}
&H(t) - \id_{\leq M} H^{\rm{B}}(t) 
\nonumber \\
\label{eq:difference Hamiltonian and Bogoliubov Hamiltonian a}
&\quad = \id_{> M} \left( \int dx\, b^{*}_x h(t) b_x + \mathcal{N}_a \right)  
\\
\label{eq:difference Hamiltonian and Bogoliubov Hamiltonian b}
&\qquad +  \id_{> M} \left(
\int dx \int dk\,  K(t,k,x) \big( a^{*}_k + a_{-k}\big) b^*_x \big[ 1- N^{-1}\Number_b \big]_+^{1/2} + \text{h.c.}
\right)  
\\
\label{eq:difference Hamiltonian and Bogoliubov Hamiltonian c}
&\qquad +  \id_{\leq M} \left(
\int dx \int dk\,  K(t,k,x) \big( a^{*}_k + a_{-k}\big) b^*_x \left(  \big[ 1- N^{-1}\Number_b \big]_+^{1/2} - 1 \right) + \text{h.c.}
\right)  
\\
\label{eq:difference Hamiltonian and Bogoliubov Hamiltonian d}
&\qquad + N^{-1/2} \int dx\, b^{*}_x \Big( q(t) \widehat{\Phi} q(t) - \scp{\psi_t}{\widehat{\Phi} \psi_t}_{L^2(\mathbb{R}^3)} \Big) b_x   .
\end{align}
\end{subequations}
The contribution from \eqref{eq:difference Hamiltonian and Bogoliubov Hamiltonian a} vanishes because the operators in the brackets leave the total number of excitations invariant. We, moreover, have
\begin{align}
\eqref{eq:difference Hamiltonian and Bogoliubov Hamiltonian b}  \id_{\leq M}
&= 
\id_{> M} 
\int dx \int dk\,  K(t,k,x)  \left( a^{*}_k + a_{-k} \right)  b^*_x \big[ 1- N^{-1}\Number_b \big]_+^{1/2}   \id_{M-1 \leq \mathcal{N}  \leq M}  .
\end{align}
This leads to
\begin{subequations}
\begin{align}
& \frac{1}{2} \frac{d}{dt}\norm{\chi(t) - \chbm(t)}^2
\nonumber \\
\label{eq:time derivative norm estimate 1}
&\quad =  \Im \scp{\chi(t)}{\id_{> M} 
\int dx \int dk\,  K(t,k,x) \left( a^{*}_k + a_{-k} \right)  b^*_x \big[ 1- N^{-1}\Number_b \big]_+^{1/2}   \id_{M-1 \leq \mathcal{N}  \leq M}  \chbm(t)}
\\
\begin{split}
\label{eq:time derivative norm estimate 2}
&\qquad + 
 \Im \bigg\langle \chi(t) , \id_{\leq M} \bigg(
\int dx \int dk\,  K(t,k,x) \big( a^{*}_k + a_{-k}\big) b^*_x 
\\
&\qquad \qquad  \qquad \qquad \times 
\left(  \big[ 1- N^{-1}\Number_b \big]_+^{1/2} - 1 \right) + \text{h.c.}
\bigg)   \chbm(t) \bigg\rangle
\end{split}
\\
\label{eq:time derivative norm estimate 3}
&\qquad +
 N^{-1/2} \Im \scp{\chi(t)}{\int dx\, b^{*}_x \Big( q(t) \widehat{\Phi} q(t) - \scp{\psi_t}{\widehat{\Phi} \psi_t}_{L^2(\mathbb{R}^3)} \Big) b_x     \chbm(t)} .
\end{align}
\end{subequations}
In the following we estimate each term separately.

\paragraph{The term \eqref{eq:time derivative norm estimate 1}}
Using Lemma \ref{lemma:preliminary estimates} we bound the first term by
\begin{align}
\abs{\eqref{eq:time derivative norm estimate 1}}
&\leq C \norm{\psi_t}_{H^1} \norm{\left( \mathcal{N}_b + T_b \right)^{1/2} \id_{>M} \chi(t)}
\norm{\left( \mathcal{N}_a + 1 \right)^{1/2} \big[ 1- N^{-1}\Number_b \big]_+^{1/2}   \id_{M-1 \leq \mathcal{N}  \leq M}  \chbm(t)}
\nonumber \\
&\leq C \norm{\psi_t}_{H^1} \norm{\left( \mathcal{N}_b + T_b \right)^{1/2}\chi(t)}
\norm{\left( \mathcal{N}_a + 1 \right)^{1/2}    \id_{M-1 \leq \mathcal{N}  \leq M}  \chbm(t)}
\nonumber \\
&\leq C M^{-1} \norm{\psi_t}_{H^1} \norm{\left( \mathcal{N}_b + T_b \right)^{1/2}\chi(t)}
\norm{\left( \mathcal{N} + 1 \right)^{3/2}    \chbm(t)} .
\end{align}
By means of \eqref{eq:initial condition moment assumption}, \eqref{eq:estimate number operator many-body dynamics} and Lemma \ref{lemma:well-posedness truncated Bogoliubov dynamics} we get
\begin{align}
\abs{\eqref{eq:time derivative norm estimate 1}}
&\leq C e^{C f(t)} \norm{\psi_t}_{H^1(\mathbb{R}^3)} M^{-3/8} .
\end{align}

\paragraph{The term \eqref{eq:time derivative norm estimate 2}}

Note that
\begin{subequations}
\begin{align}
&\abs{\eqref{eq:time derivative norm estimate 2}}
\nonumber \\
\label{eq:time derivative norm estimate 2a}
&\quad \leq 
\bigg| \bigg\langle \chi(t) ,\id_{\leq M} \bigg(
 \int dk\, \scp{\psi_t}{G_{\cdot}(k) \psi_t} \big( a^{*}_k + a_{-k}\big) b^* (\psi_t) 
\left(  \big[ 1- N^{-1}\Number_b \big]_+^{1/2} - 1 \right) + \text{h.c.}
\bigg)   \chbm(t) \bigg\rangle  \bigg|
\\
\label{eq:time derivative norm estimate 2b}
&\qquad + 
\abs{\scp{\chi(t)}{\id_{\leq M} 
\int dx \int dk \,  \psi_t(x) G_x(k) \big( a^{*}_k + a_{-k}\big) b^*_x \left(  \big[ 1- N^{-1}\Number_b \big]_+^{1/2} - 1 \right)   \chbm(t)}}
\\
\label{eq:time derivative norm estimate 2c}
&\qquad + 
\abs{\scp{\int dx \int dk \,  \psi_t(x) G_x(k) \big( a^{*}_k + a_{-k}\big) b^*_x \left(  \big[ 1- N^{-1}\Number_b \big]_+^{1/2} - 1 \right) \id_{\leq M} \, \chi(t)}{  \chbm(t)} } 
\end{align}
\end{subequations}
follows directly from the definition of $K$.
Due to \eqref{eq: bound for annihialtion and creation operators} and the shifting property of the number operator we have
\begin{align}
\abs{\eqref{eq:time derivative norm estimate 2a}}
&\leq C \norm{\psi_t}_{H^1(\mathbb{R}^3)} \norm{\left(\mathcal{N} + 1 \right)^{1/2} \chi(t)}
\norm{\left( \mathcal{N} + 1 \right)^{1/2} \left(  \big[ 1- N^{-1} ( \Number_b - 1 ) \big]_+^{1/2} - 1 \right)  \id_{\mathcal{N}_b \geq 1} \chbm(t)}
\nonumber \\
&\quad +
C \norm{\psi_t}_{H^1(\mathbb{R}^3)} \norm{\left(\mathcal{N} + 1 \right)^{1/2} \chi(t)}
\norm{\left( \mathcal{N} + 1 \right)^{1/2} \left(  \big[ 1- N^{-1}\Number_b \big]_+^{1/2} - 1 \right) \chbm(t)} .
\end{align}
Using $\big(  \big[ 1- N^{-1} ( \Number_b - 1 ) \big]_+^{1/2} - 1 \big)  \id_{\mathcal{N}_b \geq 1} \leq C N^{-1} \mathcal{N}_b$ and
$\left[1 - N^{-1} \mathcal{N}_b \right]_{+}^{\frac{1}{2}} - 1 \leq C N^{-1} \mathcal{N}_b$, which is a consequence of $\left[ 1 - x \right]_{+}^{\frac{1}{2}} - 1 \leq x$ for all $x \geq 0$ and the spectral calculus, we obtain
\begin{align}
\abs{\eqref{eq:time derivative norm estimate 2a}}
&\leq C N^{-1} \norm{\psi_t}_{H^1(\mathbb{R}^3)} \norm{\left(\mathcal{N} + 1 \right)^{1/2} \chi(t)}
\norm{\left( \mathcal{N} + 1 \right)^{3/2}  \chbm(t)} .
\end{align}
By means of Lemma \ref{lemma:preliminary estimates} we estimate 
\begin{align}
\abs{\eqref{eq:time derivative norm estimate 2b}}
&\leq  C \norm{\psi_t}_{H^1(\mathbb{R}^3)} \norm{\left( \mathcal{N}_b + T_b \right)^{1/2} \chi(t)}
\norm{\left( \mathcal{N}_a + 1 \right)^{1/2} \left(  \big[ 1- N^{-1}\Number_b \big]_+^{1/2} - 1 \right) \chbm(t)}
\nonumber \\
&\leq C N^{-1} \norm{\psi_t}_{H^1(\mathbb{R}^3)} \norm{\left( \mathcal{N}_b + T_b \right)^{1/2} \chi(t)}
\norm{\left( \mathcal{N} + 1 \right)^{3/2}  \chbm(t)} .
\end{align}
Before estimating $\eqref{eq:time derivative norm estimate 2c}$ we shift the term involving the number operator to the right hand side of the scalar product and split the integral in $k$ by means of a cutoff parameter $\Lambda_1 \geq 1$. Applying \eqref{eq:preliminary estimates quadratic term 1} and \eqref{eq:preliminary estimates quadratic term 2} with $\Lambda = \Lambda_1$ then leads to
\begin{align}
&\abs{\eqref{eq:time derivative norm estimate 2c}}
\nonumber \\
&\quad \leq C \norm{\psi_t}_{H^1(\mathbb{R}^3)} 
\norm{\left( \mathcal{N}_a + 1 \right)^{1/2} \id_{\leq M} \chi(t)}
\bigg[  \Lambda_1^{1/2} 
\norm{\mathcal{N}_b ^{1/2} \left(  \big[ 1- N^{-1} (\Number_b - 1) \big]_+^{1/2} - 1 \right) \id_{\mathcal{N}_b \geq 1} \chbm(t)}
\nonumber \\
&\qquad +
 \Lambda_1^{-1/2} 
\norm{\left( \mathcal{N}_b + T_b \right)^{1/2} \left(  \big[ 1- N^{-1} (\Number_b - 1) \big]_+^{1/2} - 1 \right) \id_{\mathcal{N}_b \geq 1} \chbm(t)}
\bigg]
\nonumber \\
&\quad \leq C \norm{\psi_t}_{H^1(\mathbb{R}^3)} \norm{\left( \mathcal{N} + 1 \right)^{1/2}  \chi(t)}
\bigg[  N^{-1} \Lambda_1^{1/2} \norm{\left( \mathcal{N} + 1 \right)^{3/2} \chbm(t)}
+ \Lambda_1^{-1/2}  \norm{\left( \mathcal{N}_b + T_b \right)^{1/2}  \chbm(t)}
\bigg] .
\end{align}
Summing up, we get
\begin{align}
\abs{\eqref{eq:time derivative norm estimate 2} }
&\leq C \norm{\psi_t}_{H^1(\mathbb{R}^3)} N^{-1}  \Lambda_1^{1/2}  \norm{\left( \mathcal{N} + T_b + 1 \right)^{1/2} \chi(t)} \norm{\left( \mathcal{N} + 1 \right)^{3/2} \chbm(t)}
\nonumber \\
&\quad + C \norm{\psi_t}_{H^1(\mathbb{R}^3)} \Lambda_1^{- 1/2} \norm{\left( \mathcal{N} + 1 \right)^{1/2} \chi(t)}
\norm{\left( \mathcal{N}_b + T_b \right)^{1/2} \chbm(t)} .
\end{align}
Using \eqref{eq:initial condition moment assumption}, \eqref{eq:estimate number operator many-body dynamics}, as well as Lemma \ref{lemma:well-posedness truncated Bogoliubov dynamics} and setting $\Lambda_1 = N M^{- 5/8}$ leads to
\begin{align}
\abs{\eqref{eq:time derivative norm estimate 2} }
&\leq C e^{C f(t)}  \norm{\psi_t}_{H^1(\mathbb{R}^3)} \left(  N^{-1} M^{5/8}  \Lambda_1^{1/2}  
 +  \Lambda_1^{- 1/2} \right) 
= C e^{C f(t)} \norm{\psi_t}_{H^1(\mathbb{R}^3)} N^{-1/2} M^{5/16}  .
\end{align}

\paragraph{The term \eqref{eq:time derivative norm estimate 3}}

By means of  $\chbm(t) = \id_{\leq M}  \chbm(t)$, \eqref{eq:preliminary estimate cubic term 1} with $\Lambda = \Lambda_2 \geq 1$, \eqref{eq:initial condition moment assumption}, \eqref{eq:estimate number operator many-body dynamics} and Lemma \ref{lemma:well-posedness truncated Bogoliubov dynamics} we get
\begin{align}
\abs{\eqref{eq:time derivative norm estimate 3}}
&\leq  N^{-1/2} \abs{\scp{\id_{\leq M + 1} \, \chi(t)}{\int dx\, b^{*}_x \Big( q(t) \widehat{\Phi} q(t) - \scp{\psi_t}{\widehat{\Phi} \psi_t} \Big) b_x   \id_{\leq M}  \chbm(t)} }
\nonumber \\
&\leq C  \norm{\psi_t}_{H^1(\mathbb{R}^3)} N^{-1/2}
\bigg[
\Lambda_2^{1/2} 
\norm{\left( \mathcal{N}_b + T_b \right)^{1/2}\chi(t)} \norm{\left( \mathcal{N}_a + 1 \right)^{1/2} \mathcal{N}_b^{1/2} \chbm(t)}
\nonumber \\
&\qquad \qquad \qquad \qquad +
\Lambda_2^{-1/2} \norm{\left( \mathcal{N}_a + 1 \right)^{1/2} \mathcal{N}_b^{1/2} \id_{\leq M + 1} \chi(t)}
\norm{\left( \mathcal{N}_b + T_b \right)^{1/2} \chbm(t)}
\bigg]
\nonumber \\
&\leq C  \norm{\psi_t}_{H^1(\mathbb{R}^3)}  N^{-1/2}
\norm{\left( \mathcal{N}_b + T_b \right)^{1/2}\chi(t)} 
\bigg[ \Lambda_2^{1/2} 
\norm{\left( \mathcal{N}_a + 1 \right)^{1/2} \mathcal{N}_b^{1/2} \chbm(t)}
\nonumber \\
&\qquad \qquad \qquad \qquad \qquad \qquad  
\qquad \qquad  \quad + \Lambda_2^{-1/2}    M^{1/2} 
\norm{\left( \mathcal{N}_b + T_b \right)^{1/2} \chbm(t)}
\bigg] 
\nonumber \\
&\leq C  e^{C f(t)}  \norm{\psi_t}_{H^1(\mathbb{R}^3)} N^{-1/2}
\left(
 M^{1/4} \Lambda_2^{1/2}    
+
M^{1/2}  \Lambda_2^{-1/2}    
\right) .
\end{align}
Choosing $\Lambda_2 =  M^{1/4}$ leads to
\begin{align}
\abs{\eqref{eq:time derivative norm estimate 3}}
&\leq C  e^{C f(t)} \norm{\psi_t}_{H^1(\mathbb{R}^3)}  N^{-1/2} M^{3/8} .
\end{align}
If we collect the estimates and use \eqref{eq:estimate Pekar energy} to bound the $H^1$-norm of the condensate wave function we obtain
\begin{align}
\frac{d}{dt} \norm{\chi(t)  - \chbm(t)}^2
&\leq  C  e^{C f(t)} \left( M^{-3/8} +   N^{-1/2} M^{3/8}  \right) .
\end{align}
Since
\begin{align}
\norm{\chi(0)  - \chbm(0)}
&\leq   \norm{\chi(0) - \chi}  
+  \norm{\chi -  \chbm(0)}
\leq  \norm{\id_{\mathcal{N}_b > N} \chi}
+  \norm{\id_{\mathcal{N} > M} \chi}
\leq  C \left( N^{- 3/2} + M^{- 3/2} \right) 
\end{align}
we get
\begin{align}
\norm{\chi(t)  - \chbm(t)}^2
&\leq  C  e^{C f(t)} \left( M^{-3/8} +   N^{-1/2} M^{3/8} + N^{-3} \right) 
\end{align}
by Duhamel's formula. Plugging this estimate and \eqref{eq:Bog dynamics with and without cutoff norm estimate} into \eqref{eq:norm-approximation first step in the proof} and choosing $M= N^{2/3}$ concludes the proof.
\end{proof}

\appendix

\section{Proof of Lemma \ref{lemma:preliminary estimates} and Lemma \ref{lemma: cubic interaction}}
\label{appendix:estimates for the interaction}

\begin{proof}[Proof of Lemma \ref{lemma:preliminary estimates}]
Using the commutator method of Lieb and Yamazaki \cite{LY1958}, i.e. 
\begin{align}
\label{eq:commutator with G-x}
\big( 1 + \abs{k}^2 \big) G_x(k) =  
G_x(k) +  (2 \pi)^{-1}   \left[ k \cdot i \nabla_x , G_x(k) \right]
\end{align} 
and integration by parts,
we write the left hand side of the first inequality as
\begin{align}
&\scp{\chi}{\int dx \, \int_{\abs{k} \geq \Lambda} dk \, \psi_t(x) G_x(k) \left( a_k^* + a_{-k} \right) b_x^* \Psi}
\nonumber \\
&\quad =
\int dx \, \int_{\abs{k} \geq \Lambda} dk \, \frac{1}{1 + k^2} \scp{b_x \chi}{ G_x(k) \left( a_k^* + a_{-k} \right) \psi_t(x)  \Psi}
\nonumber \\
&\qquad + \frac{1}{2 \pi}
\int dx \, \int_{\abs{k} \geq \Lambda} dk \, \frac{k}{1 + k^2} \scp{i \nabla_x b_x \chi}{ G_x(k) \left( a_k^* + a_{-k} \right) \psi_t(x)  \Psi}
\nonumber \\
&\qquad - \frac{1}{2 \pi}
\int dx \, \int_{\abs{k} \geq \Lambda} dk \, \frac{k}{1 + k^2} \scp{ b_x \chi}{ G_x(k) \left( a_k^* + a_{-k} \right) i \nabla_x \psi_t(x)  \Psi} .
\end{align}
By means of the Cauchy--Schwarz inequality we obtain
\begin{align}
&\abs{\scp{\chi}{\int dx \, \int_{\abs{k} \geq \Lambda} dk \, \psi_t(x) G_x(k) \left( a_k^* + a_{-k} \right) b_x^* \Psi}}
\nonumber \\
&\quad \leq C \norm{\psi_t}_{H^1(\mathbb{R}^3)} \norm{\id_{\abs{\cdot} \geq \Lambda}  (1 + \abs{\cdot}^2)^{-1}}_{L^2(\mathbb{R}^3)} \norm{\left( \mathcal{N}_a + 1 \right)^{1/2} \Psi} \left( \norm{\mathcal{N}_b^{1/2} \chi} 
+ \norm{T_b^{1/2} \chi} \right) .
\end{align}
Together with $\norm{\id_{\abs{\cdot} \geq \Lambda}  (1 + \abs{\cdot}^2)^{-1}}_{L^2(\mathbb{R}^3)} \leq \sqrt{4 \pi/ \Lambda} $ this shows \eqref{eq:preliminary estimates quadratic term 1}.
Using
\begin{align}
\norm{\int_{\abs{k} \leq \Lambda} dk \, G_x(k) a_k^* \psi_t(x) \Psi}^2
&=
\norm{\int_{\abs{k} \leq \Lambda} dk \, G_x(k) a_{-k} \psi_t(x) \Psi}^2 
+ 4 \pi \Lambda \abs{\psi_t(x)}^2 \norm{\Psi}^2 
\end{align}
and the Cauchy--Schwarz inequality we estimate
\begin{align}
&\abs{\scp{\chi}{\int dx \, \int_{\abs{k} \leq \Lambda} dk \, \psi_t(x) G_x(k) \left( a_k^* + a_{-k} \right) b_x^* \Psi} }
\nonumber \\
&\quad \leq 
\int dx \, \norm{b_x \chi} 
\left( 
\norm{\int_{\abs{k} \leq \Lambda} dk \,  G_x(k)  a_{-k}  \psi_t(x) \Psi}
+
\norm{\int_{\abs{k} \leq \Lambda} dk \,  G_x(k)  a_{k}^*  \psi_t(x) \Psi}
\right)
\nonumber \\
&\quad \leq C \norm{\mathcal{N}_b^{1/2} \chi}
\left[
\left( \int \, dx \norm{\int_{\abs{k} \leq \Lambda} dk \,  G_x(k)  a_{-k}  \psi_t(x) \Psi}^2 \right)^{1/2}
+
\Lambda^{1/2} \norm{\psi_t}_{L^2(\mathbb{R}^3)} \norm{\Psi}
\right]  .
\end{align}
Using that \cite[Lemma 10]{FS2014} implies
$
a^* \left( \id_{\abs{\cdot} \leq \Lambda} G_x \right) 
a \left( \id_{\abs{\cdot} \leq \Lambda} G_x \right)
\leq C \left( 1 - \Delta_x \right) \mathcal{N}_a
$( see \cite[(6.8)]{LMS2021})  we estimate
\begin{align}
\norm{\int_{\abs{k} \leq \Lambda} dk \,  G_x(k)  a_{-k}  \psi_t(x) \Psi}
&\leq C \norm{\left( 1 - \Delta_x \right)^{1/2} \psi_t(x) \mathcal{N}_a^{1/2} \Psi} .
\end{align}
Altogether this shows \eqref{eq:preliminary estimates quadratic term 2}. Using $K(t,k,x) = \psi_t(x) \left( G_x(k) - \scp{\psi_t}{G_{\cdot}(k) \psi_t}_{L^2(\mathbb{R}^3)} \right)$ we bound the left hand side of \eqref{eq:preliminary estimates quadratic term 3} by
\begin{align}
&\abs{\scp{\chi}{ \int dx \, \int dk \, K(t,k,x) \left(a^*_k + a_{-k} \right) b^*_x  \Psi}}
\nonumber \\
&\quad \leq 
\abs{\scp{\chi}{ \int dx \, \int dk \, \psi_t(x) \scp{\psi_t}{G_{\cdot}(k) \psi_t}_{L^2(\mathbb{R}^3)} \left(a^*_k + a_{-k} \right) b^*_x \Psi}}
\nonumber \\
&\qquad +
\abs{\scp{\chi}{ \int dx \, \int dk \, \psi_t(x) G_x(k) \left(a^*_k + a_{-k} \right) b^*_x \Psi}} .
\end{align}
Due to \eqref{eq:estimate expectation value of G} we have that
\begin{align}
\label{eq:preliminary estimates auxiliary estimates 1}
\abs{\scp{\chi}{ \int dx \, \int dk \, \psi_t(x) \scp{\psi_t}{G_{\cdot}(k) \psi_t}_{L^2(\mathbb{R}^3)} \left(a^*_k + a_{-k} \right) b^*_x   \Psi}} 
&\leq C \norm{\psi_t}_{H^1} \norm{\mathcal{N}_b^{1/2} \chi} \norm{\left( \mathcal{N}_a + 1 \right)^{1/2} \Psi} .
\end{align}
Together with the previous estimates this shows \eqref{eq:preliminary estimates quadratic term 3}.
\end{proof}

\begin{proof}[Proof of Lemma \ref{lemma: cubic interaction}]
By means of 
\begin{align}
\int dx \, b_x^* \scp{\psi_t}{\hat{\Phi} \psi_t}_{L^2(\mathbb{R}^3)} b_x = \mathcal{N}_b \left( a \left( \scp{\psi_t}{G_{\cdot} \psi_t}_{L^2(\mathbb{R}^3)} \right) + a^* \left( \scp{\psi_t}{G_{\cdot} \psi_t}_{L^2(\mathbb{R}^3)} \right)  \right) 
\end{align}
and
\eqref{eq:estimate expectation value of G} we get
\begin{align}
\abs{\scp{\chi}{\int dx \, b_x^* \scp{\psi_t}{\widehat{\Phi} \psi_t}_{L^2(\mathbb{R}^3)} b_x \Psi} } \leq C \norm{\psi_t}_{H^1(\mathbb{R}^3)} \norm{\mathcal{N}_b^{\frac{1}{2}} \chi}
\norm{\left( \mathcal{N}_a + 1 \right)^{\frac{1}{2}} \mathcal{N}_b^{\frac{1}{2}} \Psi} .
\end{align}
Moreover, note that
\begin{align}
\label{eq:preliminary estimate cubic term intermediate 1}
\abs{\scp{\chi}{\int dx \, b_x^* q(t) \widehat{\Phi} q(t) b_x \Psi}}
&\leq \sum_{\sharp \in \{\cdot, *\} } \abs{\scp{\chi}{\int dx \, b_x^* q(t)  a^{\sharp}
\left( \id_{\abs{\cdot} \leq \Lambda} G_{\cdot} \right)
 q(t) b_x \Psi}}
\\
\label{eq:preliminary estimate cubic term intermediate 2}
&\quad + \sum_{\sharp \in \{\cdot, *\} } 
\abs{\scp{\chi}{\int dx \, b_x^* q(t) 
 a^{\sharp} \left( \id_{\abs{\cdot} \geq \Lambda} G_{\cdot} \right)
q(t) b_x \Psi}} .
\end{align}
Using
\begin{align}
&\int dx\, b^{*}_x  q(t) a^{\sharp}
\left( \id_{\abs{\cdot} \leq \Lambda} G_{\cdot} \right) q(t) b_x 
\nonumber \\
&\quad = b^* \left( \psi_t \right) a^{\sharp}
\left( \id_{\abs{\cdot} \leq \Lambda} \scp{\psi_t}{G_{\cdot} \psi_t}_{L^2(\mathbb{R}^3)} \right)  b \left( \psi_t \right) 
+ \int dx \, b_x^* a^{\sharp}
\left( \id_{\abs{\cdot} \leq \Lambda} G_x \right) b_x 
\nonumber \\
&\qquad  - b^* \left( \psi_t \right) \int dx \, \overline{\psi_t(x)} a^{\sharp}
\left( \id_{\abs{\cdot} \leq \Lambda} G_x \right) b_x
- \int dx \, b_x^* a^{\sharp}
\left( \id_{\abs{\cdot} \leq \Lambda} G_x \right) \psi_t(x) b \left( \psi_t \right) ,
\end{align}
as well as \eqref{eq: bound for annihialtion and creation operators} and \eqref{eq:bound for G}
we obtain
\begin{align}
\abs{\eqref{eq:preliminary estimate cubic term intermediate 1}}
&\leq C \sup_{x \in \mathbb{R}^3} \norm{\id_{\abs{\cdot} \leq \Lambda} G_x}_{L^2(\mathbb{R}^3)} \norm{\mathcal{N}_b^{\frac{1}{2}} \chi} \norm{\left( \mathcal{N}_a + 1 \right)^{\frac{1}{2}} \mathcal{N}_b^{\frac{1}{2}} \Psi}
\leq C \Lambda^{\frac{1}{2}}  \norm{\mathcal{N}_b^{\frac{1}{2}} \chi} \norm{\left( \mathcal{N}_a + 1 \right)^{\frac{1}{2}} \mathcal{N}_b^{\frac{1}{2}} \Psi} .
\end{align}
Similarly, 
\begin{align}
&\abs{\scp{\chi}{\int dx \, b_x^* q(t) 
 a^{\sharp} \left( \id_{\abs{\cdot} \geq \Lambda} G_{\cdot} \right)
q(t) b_x \Psi}} 
\nonumber \\
&\quad \leq \abs{\scp{b (\psi_t) \chi}{a^{\sharp} \left( \id_{\abs{\cdot} \geq \Lambda} \scp{\psi_t}{G_{\cdot} \psi_t}_{L^2(\mathbb{R}^3)} \right) b (\psi_t) \Psi}}
+
\abs{ \int dx \,  \scp{b_x \chi}{a^{\sharp} \left( \id_{\abs{\cdot} \geq \Lambda} G_x \right) b_x \Psi }}
\nonumber \\
&\qquad +
\abs{\int dx \,  \scp{\psi_t(x) b (\psi_t) \chi}{a^{\sharp} \left( \id_{\abs{\cdot} \geq \Lambda} G_x \right) b_x \Psi}}
+
\abs{\int dx \,  \scp{b_x \chi}{a^{\sharp} \left( \id_{\abs{\cdot} \geq \Lambda} G_x \right) \psi_t(x) b(\psi_t) \Psi}} .
\end{align}
By means of \eqref{eq:commutator with G-x} and integration by parts we get
\begin{align}
\abs{\eqref{eq:preliminary estimate cubic term intermediate 2}}
&\leq C \norm{\psi_t}_{H^1(\mathbb{R}^3)}
\norm{\big(1 + \abs{\cdot}^2 \big)^{-1} \id_{\abs{\cdot} \geq \Lambda} }_{L^2(\mathbb{R}^3)}
\Big( 
\norm{\left( \mathcal{N}_b + T_b \right)^{\frac{1}{2}} \chi}
\norm{\left( \mathcal{N}_a + 1 \right)^{\frac{1}{2}} \mathcal{N}_b^{\frac{1}{2}} \Psi}
\nonumber \\
&\qquad \qquad \qquad  \qquad \qquad \qquad \qquad \qquad \qquad  \quad 
+ \norm{\left( \mathcal{N}_a + 1 \right)^{\frac{1}{2}} \mathcal{N}_b^{\frac{1}{2}} \chi} \norm{\left( \mathcal{N}_b + T_b \right)^{\frac{1}{2}} \Psi}
\Big)
\nonumber \\
&\leq C \norm{\psi_t}_{H^1(\mathbb{R}^3)}
\Lambda^{- \frac{1}{2}}
\left( 
\norm{\left( \mathcal{N}_b + T_b \right)^{\frac{1}{2}} \chi}
\norm{\left( \mathcal{N}_a + 1 \right)^{\frac{1}{2}} \mathcal{N}_b^{\frac{1}{2}} \Psi}
+ \norm{\left( \mathcal{N}_a + 1 \right)^{\frac{1}{2}} \mathcal{N}_b^{\frac{1}{2}} \chi} \norm{\left( \mathcal{N}_b + T_b \right)^{\frac{1}{2}} \Psi}
\right) .
\end{align}
In total, this shows \eqref{eq:preliminary estimate cubic term 1}. Setting $\Lambda = 1$ and $\chi = \Psi$ in \eqref{eq:preliminary estimate cubic term 1} and applying Young's inequality for products leads to \eqref{eq:preliminary estimate cubic term 2}.
\end{proof}

\section{Proof of Proposition \ref{proposition:solution theory for LP equation}}

\label{section:solution theory for LP equation}

Proposition \ref{proposition:solution theory for LP equation} is, except of bound for $\norm{\psi_t}_{H^3(\mathbb{R}^3)}$, a direct consequence \cite[Lemma~2.1 and Proposition~C.2]{FG2017}. In order to obtain the missing estimate we modify the proof of \cite[Proposition 2.2]{FG2017}. There, the $H^4$-norm of $\psi_t$ was estimated by means a functional which is better controllable during the time evolution than $\norm{\psi_t}_{H^4(\mathbb{R}^3)}$.
We will rely on the following results.
\begin{proposition}[Part of Proposition C.2 in \cite{FG2017}]
\label{proposition: Proposition C.2 in FG2017}
If $\alpha = 1$, $(\psi, \varphi) \in H^2(\mathbb{R}^3) \times L_1^2(\mathbb{R}^3)$, then $(\psi_t, \varphi_t) \in H^2(\mathbb{R}^3) \times L_1^2(\mathbb{R}^3)$ for all $t \in \mathbb{R}$ and there exists a constant $C> 0$ depending only on the initial data such that
\begin{align}
\norm{\psi_t}_{H^2(\mathbb{R}^3)} \leq C \left( 1 + \abs{t} \right) 
\quad \text{and} \quad
\norm{\varphi_t}_{L_1^2(\mathbb{R}^3)} \leq C \left( 1 +  \abs{t} \right) .
\end{align}
If, in addition, $\varphi \in L_2^2(\mathbb{R}^3)$ then $\varphi_t \in L_2^2(\mathbb{R}^3)$ for all $t \in \mathbb{R}$ and there exists a constant $C> 0$ depending only on the initial data such that 
\begin{align}
\label{eq:prelimary lemma solution theory L-2-2 bound for varphi}
\norm{\varphi_t}_{L_2^2(\mathbb{R}^3)}
&\leq C \left( 1 + \abs{t}^3 \right) .
\end{align}
\end{proposition}

\begin{lemma}
\label{lemma:preliminary estimates for solution theory}
There exists a constant $C > 0$ such that 
\begin{align}
\norm{\partial^{\beta} \Phi_{\varphi_t}}_{L^{\infty}(\mathbb{R}^3)} &\leq  C \norm{\varphi_t}_{L_{\abs{\beta} +1}^2(\mathbb{R}^3)} ,
\quad 
\norm{\partial^{\beta} \Phi_{\dot{\varphi}_t}}_{L^{\infty}(\mathbb{R}^3)} \leq  C \norm{\varphi_t}_{L_{\abs{\beta} +1}^2(\mathbb{R}^3)}
\end{align}
for all $\beta \in \mathbb{N}_{0}^3$ and such that 
\begin{align}
\label{eq:preliminary estimates for solution theory bound for energy}
1 \leq - \Delta + \Phi_{\varphi_t} + C \left( \mathcal{E}[\psi, \varphi] + C \right) .
\end{align}
\end{lemma}

\begin{proof}[Proof of Lemma \ref{lemma:preliminary estimates for solution theory}]
If we insert $\big( 1 + \abs{\cdot}^2 \big)^{ \frac{1 + \abs{\beta}}{2}}$ and it's inverse to the right hand side of $\partial_x^{\beta} \Phi_{\varphi_t}(x) = 2 \Re \scp{\abs{\cdot}^{-1} \partial_x^{\beta} e^{- 2 \pi i x \cdot}}{\varphi_t}_{L^2(\mathbb{R}^3)} $ and apply the Cauchy--Schwarz inequality we obtain the first estimate. The second inequality is derived by similar means because $\Phi_{\dot{\varphi}_t}=  2 \Im \scp{\abs{\cdot}^{-1} e^{- 2 \pi i x \cdot}}{\varphi_t}_{L^2(\mathbb{R}^3)}$. By  \eqref{eq:estimate expectation value of G} and the Cauchy--Schwarz inequality we get
\begin{align}
\abs{\scp{\xi}{\Phi_{\varphi_t} \xi}_{L^2(\mathbb{R}^3)}}
&\leq C \norm{\xi}_{H^1(\mathbb{R}^3)} \norm{\xi}_{L^2(\mathbb{R}^3)} \norm{\varphi_t}_{L^2(\mathbb{R}^3)}
\leq \frac{1}{2} \norm{\xi}_{H^1(\mathbb{R}^3)}^2 + C \norm{\varphi_t}_{L^2(\mathbb{R}^3)}^2 \norm{\xi}_{L^2(\mathbb{R}^3)}^2 
\end{align}
for $\xi \in H^1(\mathbb{R}^3)$. Together with \eqref{eq:estimate Pekar energy} this leads to \eqref{eq:preliminary estimates for solution theory bound for energy}.
\end{proof}

\begin{proof}[Bound for $\norm{\psi_t}_{H^3(\mathbb{R}^3)}$]

The local well-posedness of solution in $H^3(\mathbb{R}^3) \times L_2^2(\mathbb{R}^3)$ can be shown by a standard fixed-point argument. In order to derive a bound on the $H^3$-norm of $\psi_t$ we define the functional
\begin{align}
\mathcal{E}^{(3)}[\psi_t, \varphi_t]
&= \norm{\left( - \Delta + \Phi_{\varphi_t} + M \right)^{3/2} \psi_t}_{L^2(\mathbb{R}^3)}^2 ,
\end{align}
where $M \geq 1$ is a constant depending only on $\mathcal{E}[\psi, \varphi]$ such that $1 \leq - \Delta + \Phi_{\varphi_t} + M$. Note that the existence of $M$ is guaranteed by \eqref{eq:preliminary estimates for solution theory bound for energy}. The functional satisfies the inequalities 
\begin{align}
\label{eq:third moment of the energy comparison to Sobolev norm}
\abs{\mathcal{E}^{(3)}[\psi_t, \varphi_t]
- \norm{\left( - \Delta +M \right)^{3/2} \psi_t}_{L^2(\mathbb{R}^3)}^2}
&\leq \frac{1}{2} \mathcal{E}^{(3)}[\psi_t, \varphi_t] +  C M^2 \left( 1 +  \abs{t}^6 \right)  ,
\\
\label{eq:third moment of the energy time derivative}
\sqrt{\mathcal{E}^{(3)}[\psi_t, \varphi_t]} 
&\leq \sqrt{\mathcal{E}^{(3)}[\psi, \varphi]} + C M \left( 1 + \abs{t}^4 \right) 
\end{align}
with $C > 0$ depending only on the initial data.
Combining the estimates let us obtain
\begin{align}
\norm{\psi_t}_{H^3(\mathbb{R}^3)}
&\leq \norm{\left( - \Delta + M \right)^{3/2} \psi_t}_{L^2(\mathbb{R}^3)}
\leq C \left(  \norm{\left( - \Delta + M \right)^{3/2} \psi}_{L^2(\mathbb{R}^3)} + M \left( 1 + \abs{t}^4 \right) \right)  .
\end{align}
This proves the second inequality in \eqref{eq:solution of the SKG equations bounds for the H-3 and L-2-2 norm}. It remains to prove the inequalities from above. Using
\begin{align}
\left( - \Delta + \Phi_{\varphi_t} + M \right)^{3} - \left( - \Delta + M \right)^{3}
&=
\left( - \Delta + \Phi_{\varphi_t} + M \right)^{2}  \Phi_{\varphi_t} 
+ \Phi_{\varphi_t}  \left( - \Delta + \Phi_{\varphi_t} + M \right)^2
\nonumber \\
&\quad 
- \Phi_{\varphi_t} \left( - \Delta + \Phi_{\varphi_t} + M \right) \Phi_{\varphi_t}
+ \left( - \Delta + M \right) \Phi_{\varphi_t} \left( - \Delta + M \right)
\end{align}
and the Cauchy--Schwarz inequality let us estimate
\begin{align}
&\abs{\mathcal{E}^{(3)}[\psi_t, \varphi_t]
- \norm{\left( - \Delta +M \right)^{3/2} \psi_t}_{L^2(\mathbb{R}^3)}^2}
\nonumber \\
&\quad \leq 
2 \sqrt{\mathcal{E}^{(3)}[\psi_t, \varphi_t]}
\norm{\left( - \Delta + \Phi_{\varphi_t} + M \right)^{1/2} \Phi_{\varphi_t} \psi_t}_{L^2(\mathbb{R}^3)}
+ \norm{\left( - \Delta + \Phi_{\varphi_t} + M \right)^{1/2} \Phi_{\varphi_t}}_{L^2(\mathbb{R}^3)}^2
\nonumber \\
&\qquad + \norm{\Phi_{\varphi_t}}_{L^{\infty}(\mathbb{R}^3)} \norm{\left( - \Delta + M \right) \psi_t}_{L^2(\mathbb{R}^3)}^2
\nonumber \\
&\leq \frac{1}{2} \mathcal{E}^{(3)}[\psi_t, \varphi_t]
+ C \norm{\left( - \Delta + \Phi_{\varphi_t} + M \right)^{1/2} \Phi_{\varphi_t} \psi_t}_{L^2(\mathbb{R}^3)}^2 
+ \norm{\Phi_{\varphi_t}}_{L^{\infty}(\mathbb{R}^3)} \norm{\left( - \Delta + M \right) \psi_t}_{L^2(\mathbb{R}^3)}^2 
\nonumber \\
&\leq \frac{1}{2} \mathcal{E}^{(3)}[\psi_t, \varphi_t]
+ C M^2 \Big[1 +
\norm{\nabla \Phi_{\varphi_t}}_{L^{\infty}(\mathbb{R}^3)}^2 + \norm{ \Phi_{\varphi_t}}^2_{L^{\infty}(\mathbb{R}^3)} \norm{\psi_t}_{H^1(\mathbb{R}^3)}^2
\nonumber \\
&\qquad \qquad \qquad  \qquad \qquad 
+ \norm{ \Phi_{\varphi_t}}_{L^{\infty}(\mathbb{R}^3)}^3
+ \norm{ \Phi_{\varphi_t}}_{L^{\infty}(\mathbb{R}^3)} \norm{\psi_t}^2_{H^2(\mathbb{R}^3)}
\Big] .
\end{align}
Inequality \eqref{eq:third moment of the energy comparison to Sobolev norm} then follows from
Proposition \ref{proposition: Proposition C.2 in FG2017} and Lemma 
\ref{lemma:preliminary estimates for solution theory}. Next, we estimate
\begin{align}
&\frac{d}{dt} \mathcal{E}^{(3)}[\psi_t, \varphi_t]
\nonumber \\
&\quad = 2 \Re \scp{\psi_t}{\Phi_{\dot{\varphi}_t} \left( - \Delta + \Phi_{\varphi_t} + M \right)^2 \psi_t}_{L^2(\mathbb{R}^3)}
\nonumber \\
&\qquad +
\scp{\psi_t}{\left( - \Delta + \Phi_{\varphi_t} + M \right) \Phi_{\dot{\varphi}_t} \left( - \Delta + \Phi_{\varphi_t} + M \right) \psi_t}_{L^2(\mathbb{R}^3)}
\nonumber \\
&\quad \leq  \sqrt{\mathcal{E}^{(3)}[\psi_t, \varphi_t]} 2
\left(
\norm{\left( - \Delta + \Phi_{\varphi_t} + M \right)^{1/2} \Phi_{\dot{\varphi}_t} \psi_t}_{L^2(\mathbb{R}^3)}
+
\norm{\Phi_{\dot{\varphi}_t} \left( - \Delta + \Phi_{\varphi_t} + M \right) \psi_t}_{L^2(\mathbb{R}^3)}
\right)
\nonumber \\
&\quad \leq \sqrt{\mathcal{E}^{(3)}[\psi_t, \varphi_t]} C M \left(
\norm{\nabla \Phi_{\dot{\varphi}_t}}_{L^{\infty}(\mathbb{R}^3)} 
+ \norm{\Phi_{\dot{\varphi}_t}}_{L^{\infty}(\mathbb{R}^3)}  
\left( \norm{\psi_t}_{H^2(\mathbb{R}^3)}
+ \norm{\Phi_{\varphi_t}}_{L^{\infty}(\mathbb{R}^3)} 
\right)
\right) .
\end{align}
By means of 
Proposition \ref{proposition: Proposition C.2 in FG2017} and Lemma 
\ref{lemma:preliminary estimates for solution theory} we get
\begin{align}
\frac{d}{dt} \mathcal{E}^{(3)}[\psi_t, \varphi_t]
&\leq \sqrt{\mathcal{E}^{(3)}[\psi_t, \varphi_t]} C M  \left( 1 + \abs{t}^3 \right)  ,
\end{align}
which implies \eqref{eq:third moment of the energy time derivative}.
\end{proof}

\noindent{\it Acknowledgments.} 
N.L. gratefully acknowledges support from the Swiss National Science
Foundation through the NCCR SwissMap and funding from the European Union's
Horizon 2020 research and innovation programme under the Marie Sk\l
odowska-Curie grant agreement N\textsuperscript{o} 101024712.


{}

\end{document}